\documentclass[sigplan,screen]{acmart}

\usepackage{subcaption}
\usepackage{cases}
\usepackage{breqn}
\usepackage{algorithm}
\usepackage{algorithmic}
\newtheorem*{remark}{Remark}

\usepackage{pifont}

\usepackage{float}

\newcommand{\xmark}{\ding{55}}%
\newcommand{\cmark}{\ding{51}}%

\definecolor{babypink}{rgb}{0.96, 0.76, 0.76}

\AtBeginDocument{%
  \providecommand\BibTeX{{%
    \normalfont B\kern-0.5em{\scshape i\kern-0.25em b}\kern-0.8em\TeX}}}

\newtheorem{proposition}{Proposition}[section]

\acmConference[Conference acronym 'XX]{Make sure to enter the correct
  conference title from your rights confirmation emai}{June 03--05,
  2018}{Woodstock, NY}
\acmPrice{15.00}
\acmISBN{978-1-4503-XXXX-X/18/06}

\begin{document}

%%
%% The "title" command has an optional parameter,
%% allowing the author to define a "short title" to be used in page headers.
\title[A Framework for Privacy and Utility in Continuous LBS Interactions]{A Framework for Managing Multifaceted Privacy Leakage While Optimizing Utility in Continuous LBS Interactions}

%%
%% The "author" command and its associated commands are used to define
%% the authors and their affiliations.
%% Of note is the shared affiliation of the first two authors, and the
%% "authornote" and "authornotemark" commands
%% used to denote shared contribution to the research.
\author{Anis Bkakria}
\email{anis.bkakria@irt-systemx.fr}
\orcid{0000-0002-9758-4617}
%\authornotemark[1]
\affiliation{%
  \institution{IRT SystemX}
  \country{France}
  \postcode{43017-6221}
}
\author{Reda Yaich}
\email{reda.yaich@irt-systemx.fr}
\orcid{?}
\affiliation{%
  \institution{IRT SystemX}
  \country{France}
  \postcode{43017-6221}
}

%%
%% By default, the full list of authors will be used in the page
%% headers. Often, this list is too long, and will overlap
%% other information printed in the page headers. This command allows
%% the author to define a more concise list
%% of authors' names for this purpose.
%\renewcommand{\shortauthors}{Trovato and Tobin, et al.}

%%
%% The abstract is a short summary of the work to be presented in the
%% article.
\begin{abstract}
Privacy in Location-Based Services (LBS) has become a paramount concern with the ubiquity of mobile devices and the increasing integration of location data into various applications. This paper presents several novel contributions to advancing the understanding and management of privacy leakage in LBS. Our contributions provide a more comprehensive framework for analyzing privacy concerns across different facets of location-based interactions. Specifically, we introduce $(\epsilon, \delta)$-location privacy, $(\epsilon, \delta, \theta)$-trajectory privacy, and $(\epsilon, \delta, \theta)$-POI privacy, which offer refined mechanisms for quantifying privacy risks associated with location, trajectory, and points of interest (POI) when continuously interacting with LBS. Furthermore, we establish fundamental connections between these privacy notions, facilitating a holistic approach to privacy preservation in LBS.
Additionally, we present a lower bound analysis to evaluate the utility of the proposed privacy-preserving mechanisms, offering insights into the trade-offs between privacy protection and data utility. Finally, we instantiate our framework with the Plannar Isotopic Mechanism to demonstrate its practical applicability while ensuring optimal utility and quantifying privacy leakages across various dimensions. The evaluations provided provide a comprehensive insight into the efficacy of our framework in capturing privacy loss on location, trajectory, and points of interest while enabling quantification of the ensured accuracy.
\end{abstract}

%%
%% The code below is generated by the tool at http://dl.acm.org/ccs.cfm.
%% Please copy and paste the code instead of the example below.
%%

%\ccsdesc[500]{Computer systems organization~Embedded systems}
%\ccsdesc[300]{Computer systems organization~Redundancy}
%\ccsdesc{Computer systems organization~Robotics}
%\ccsdesc[100]{Networks~Network reliability}

%%
%% Keywords. The author(s) should pick words that accurately describe
%% the work being presented. Separate the keywords with commas.

\maketitle

\section{Introduction}
In today's interconnected world, location navigation services have become an integral part of our daily lives. With the widespread adoption of smartphones and the increasing availability of Global Positioning System (GPS) technology, users rely heavily on LBS to navigate their surroundings, find points of interest, and plan their travel routes. These services have transformed how we explore new places, discover local businesses, and efficiently reach our desired destinations. However, as location navigation services have skyrocketed, so have concerns about the privacy and security of users' personal information.

The core functionality of location navigation services requires users to disclose their precise location information to service providers. This data is essential for accurate positioning, route calculations, and providing relevant location-based recommendations. However, collecting and storing such sensitive location data raises significant privacy implications. Users may have legitimate concerns about their whereabouts being monitored, their movements being tracked, and their personal information being misused. Additionally, the aggregation and analysis of location data by service providers can lead to user profiling, targeted advertising, and potential breaches of individual privacy.

This research paper delves into the realm of continuous LBS, which involves users regularly or consistently sharing their current positions to access various services. The paper specifically addresses the security concerns arising from adversaries who have access to a series of multiple or sequential queries collected over time, leading to a robust spatiotemporal correlation. This correlation allows attackers to potentially track and monitor a particular user, enabling them to reconstruct the user's movement trajectory and deduce the specific points of interest visited by the user between two consecutive released locations.

\subsection{Related Work}
\label{sec:related_work}
A multitude of prior studies have extensively explored various techniques and methods aimed at protecting the privacy of LBS users' location information. Based on their distinguishing principles, the proposed approaches can be classified into two main categories: Obfuscation-based mechanisms (e.g., \cite{fang2014differential,kellaris2014differentially}) and cryptography-based mechanisms (e.g., \cite{shin2012privacy,stirbys2017privacy}). We primarily focus on obfuscation-based mechanisms as our approach falls within this category of privacy protection mechanisms. A full review and comparison between obfuscation-based mechanisms and cryptographic approaches is provided in \cite{jiang2021location}. 

Obfuscation-based mechanisms encompass techniques such as path obfuscation, cloaking, mix zones, and differential privacy.

Path Confusion mechanisms \cite{hoh2010achieving,eckhoff2011slotswap} leverage the spatiotemporal correlation often found in consecutive location samples. By employing a multitarget tracking algorithm \cite{gruteser2005anonymity}, which receives these location samples submitted under anonymized pseudonyms, Path Confusion ensures that it becomes challenging to reconstruct an individual's specific trajectory from the data. Unfortunately, Path Confusion mechanisms often adopt a limited adversary model, relying on a trusted proxy to introduce delays in users' queries. Consequently, there is a potential scenario where a user's location data may not be transmitted to the LBS server for an extended time. As a result, users may experience a deterioration in their service quality due to being "out-of-service" for the entirety of this period.

Cloaking methods \cite{jiang2018roblop} can also find applications in ongoing query situations aimed at protecting user trajectory privacy. This approach's fundamental principle continues to revolve around the utilization of k-anonymity \cite{sweeney2002k}, thereby preventing attackers from discerning the specific trajectory of the target user within a set of k traces. Like Path Confusion mechanisms, the Cloaking technique frequently depends on a trusted third-party anonymizer, leading to a somewhat constrained adversary model.

Mix Zones concept \cite{palanisamy2014anonymizing,palanisamy2014attack} is defined as spatial regions within the context of LBS where applications are unable to access any location information of the users present in that area. when users enter a mix zone, they adopt a new and previously unused pseudonym, enhancing their anonymity and privacy. Mix zones, while offering a level of privacy, have limitations. They may not be suitable for LBS, which requires consistent user identities, as users within mix zones are required to change their pseudonyms frequently. Furthermore, mix zones cannot support LBS when users are within the zone, as their location data cannot be reported to the LBS server during that time. 

A common limitation of Path Confusion, Cloaking, and Mix Zones-based approaches is that they are susceptible to side channel-based inferential attacks due to the potential leakage of information from various sources. In addition, they do not provide a metric to quantify the provided trade-off between privacy and utility formally. Moreover, they are unable to offer personalized privacy and utility trade-offs tailored to the specific preferences of individual users as. When privacy mechanisms use uniform parameters for all users, they fail to account for different users' distinct privacy requirements and preferences. Some users may be comfortable sharing more personal information in exchange for specific benefits, while others may prioritize maximum privacy. Without customization, these diverse preferences cannot be accommodated.

\begin{table*}[t]
\centering
\caption{Comparison of Privacy Protection Approaches for LBS}
\begin{tabular}{|l|l|l|c|c|c|c|}
\hline
\textbf{Approach} & \textbf{Mechanism Type} & \textbf{Adversary} & \multicolumn{3}{c|}{\textbf{Quantifiable Privacy}} & \textbf{Optimality} \\ \cline{4-6}
                  &                       &        \textbf{Model}                 & \textbf{Location} & \textbf{Trajectory} & \textbf{POI} & \\ \hline
Path Confusion \cite{hoh2010achieving,eckhoff2011slotswap} & Obfuscation-based & Weak & \xmark & \xmark & \xmark & - \\ \hline
Cloaking \cite{jiang2018roblop} & Obfuscation-based & Weak & \cmark & \xmark & \xmark & - \\ \hline
Mix Zones \cite{palanisamy2014anonymizing} & Obfuscation-based & Medium & \xmark & \xmark & \xmark & - \\ \hline
Geo-Indistinguishability \cite{andres2013geo} & LDP & Strong & \cmark & \xmark & \xmark & - \\ \hline
$\omega$-event DP \cite{kellaris2014differentially} & LDP & Strong & \cmark & \xmark & \xmark & - \\ \hline
$\delta$-location DP \cite{xiao2015protecting} & Adversarial LDP & Strong  & \cmark & \xmark & \xmark & \cmark \\ \hline
RDP \cite{meehan2021location} & Divergence-based Privacy & Strong & \cmark & \xmark & \xmark & - \\ \hline
This work & Adversarial LDP & Strong & \cmark & \cmark & \cmark & \cmark \\ \hline
\end{tabular}
\label{tab:comLPPM}
\end{table*}

Differential privacy-based techniques help mitigate some of the aforementioned limitations. The first effort to apply the classic differential privacy definition to continuous LBS was introduced in \cite{dwork2010differential}. The proposed method, known as "Event-level DP," seeks to protect a specific location of the user throughout a continuous flow of events (locations). Nevertheless, as demonstrated in \cite{kellaris2014differentially}, event-level DP lacks robustness when the user interacts several time from the same location with the LBS sequentially over time. The authors of \cite{kellaris2014differentially} propose \textit{$\omega$-event DP} which establishes the adjacent stream prefix to ensure the protection of updates within an \textit{$\omega$-event}. Andres et al. \cite{andres2013geo} introduce a privacy concept that expands upon differential privacy, introducing the notion of geo-indistinguishability to safeguard the precise locations of individuals when interacting with continuous LBS. Geo-indistinguishability can serve as a metric for quantifying the loss of privacy in terms of location when engaging with a LBS in a sequential manner over time. Furthermore, geo-indistinguishability provides users with the ability to delineate the trade-off between privacy and the utility captured by specifying the geographical area in which the actual user location should remain indistinguishable. Nonetheless, as demonstrated in \cite{xiao2015protecting}, the geo-indistinguishability model is susceptible to temporal correlation. This means that an adversary, by considering the release time of two consecutive (perturbed) locations in conjunction with road network constraints or the user's movement pattern, can still accurately infer the actual user's location. To overcome the previous limitation, in \cite{xiao2015protecting} authors extends the standard differential privacy to the context of continuous location sharing, with the goal of protecting the true user' position at each timestep. This proposition involves the redefinition of the concept of neighboring datasets within standard differential privacy, ultimately leading to the definition of a "$\delta$-location" as an "obfuscation locations set" of actual position samples at each timestep. Nevertheless, as highlighted by the authors, the approach in \cite{xiao2015protecting} cannot be used to quantify the privacy loss of the user trajectory.

In \cite{meehan2021location}, Meehan and Chaudhuri propose a framework that introduces a Rényi divergence-based privacy approach designed to address the privacy challenges of revealing multiple locations in location-based services. Specifically, it allows for bounding the inference of a sensitive location by incorporating the influence of previously revealed locations. By leveraging Gaussian process conditional priors, the method accounts for the inherent dependence between points in a user's location trace, ensuring that privacy loss is bounded for individual locations within a fixed radius. However, while the framework effectively mitigates the inference of specific locations, it is limited in its ability to bound or measure the inference of entire trajectories or visited points of interest (POIs).

In light of the aforementioned, several challenges need to be addressed concerning user privacy when engaging in continuous interactions with LBS. First, the development of Location Privacy-Preserving Mechanisms (LPPMs) relies on the ability to rigorously evaluate location privacy and the degradation of service quality. However, current research predominantly focuses on devising algorithms to protect privacy, with only a limited number of studies attempting to assess the efficacy of these protective measures in relation to both privacy and utility \cite{primault2018long}. Second, it is important to note that the current research and industry landscape also lacks a unified framework for the simultaneous quantification of privacy losses across locations, trajectories, and points of interest. While some efforts have been made to measure privacy loss in either of these aspects individually \cite{andres2013geo,xiao2015protecting,meehan2021location}, the absence of a holistic framework impedes a comprehensive understanding of the multifaceted nature of location-based privacy challenges. Developing a quantification framework that can evaluate privacy losses throughout a user's locations, trajectory, and in proximity to specific points of interest in a unified manner is a paramount research gap. Third, establishing a framework for the implementation of personalized privacy protection is a complex but essential endeavor, particularly for ensuring optimal user benefits from location data. In practical terms, users demonstrate distinct and evolving requirements with respect to location privacy and service utility. Furthermore, within the same user context, these requirements exhibit variability contingent upon factors such as service categories, geographical context, temporal considerations, and environmental conditions. An efficacious solution entails the development of a personalized LPPM that is intrinsically attuned to the user's distinctive preferences, encompassing their utility expectations, privacy requisites, and specific service typologies. 

Table \ref{tab:comLPPM} provides a comparative analysis of various privacy protection approaches used in Location-Based Services (LBS). The table evaluates each method based on the type of mechanism employed, the adversary model, and whether privacy is quantifiable for location, trajectory, and points of interest (POI). The adversary model is categorized into three levels: low, where a trusted third party is assumed; medium, where the adversary does not have knowledge of the LPPM's setup; and strong, where the adversary is assumed to know all aspects of the used LPPM except the user's real location. The optimality column indicates whether the privacy mechanism can achieve an optimal balance between privacy and utility. The symbol (-) is used to denote cases where optimality in terms of utility has not been demonstrated.

\subsection{Proposed Contributions}
In this paper, we introduce several contributions that significantly advance the understanding of privacy leakage on location, trajectory, and POI during continuous interactions with a LBS. These contributions build upon the concept of $\delta$-
location set, originally introduced in \cite{xiao2015protecting}. Our contributions are as follows:
\begin{enumerate}
    \item We extend the concept of $\delta$-location set \cite{xiao2015protecting} by introducing $(\epsilon, \delta)$-location privacy. This extension provides a more stringent upper bound on privacy leakage concerning locations compared to previous approaches. Specifically, we consider a strong adversary capable of identifying the set of $\delta$-obfuscation locations (Definition \ref{def:delta_los}) used to release perturbed locations.
    
    \item Building upon $(\epsilon, \delta)$-location privacy, we introduce $(\epsilon, \delta, \theta)$-trajectory privacy. This notion bounds the privacy leakage on the trajectory of the user when multiple perturbed locations are released. Moreover, we establish a direct relationship between $(\epsilon, \delta)$-location privacy and $(\epsilon, \delta, \theta)$-trajectory privacy.
    
    \item To address privacy concerns regarding POI, we introduce $(\epsilon, \delta, \theta)$-POI privacy. We demonstrate a direct correlation between $(\epsilon, \delta, \theta)$-POI privacy and $(\epsilon, \delta)$-location privacy, enabling reasoned analysis of POI privacy in the case of continuous interaction with LBS.
    
    \item We provide a lower bound to quantify the optimal utility level achievable by a mechanism that satisfies $(\epsilon, \delta)$-location privacy. This bound serves as a benchmark for assessing the effectiveness of privacy-preserving mechanisms.
    
    \item Finally, we demonstrate the practical applicability of our framework by instantiating it with the Plannar Isotopic Mechanism \cite{xiao2015protecting}. This instantiation ensures optimal utility while quantifying the privacy leakage on locations, trajectory, and visited POIs through the formal links established within our framework.
\end{enumerate}
We strongly believe that our contributions meaningfully enhance the understanding and practical implementation of privacy-preserving mechanisms in LBS, offering a comprehensive framework for balancing utility and multifaceted  privacy leakages.
\subsection{Paper Organization}
Section \ref{sec:preliminaries} introduces the foundational concepts underlying our framework. In Section \ref{sec:system_arch}, we outline the system architecture and adversarial model. Section \ref{sec:trivial_model} examines a preliminary solution and its constraints. Our proposed framework is detailed in Section \ref{sec:the_proposed_model}, followed by an instantiation in Section \ref{sec:instantiation}. Section \ref{sec:evaluation} presents the evaluation results of our framework. Finally, Section \ref{sec:conclusion} summarizes the findings and concludes the paper.
%
%\begin{definition}[Differential Privacy]

%\end{definition}

\section{Preliminaries}
\label{sec:preliminaries}
\subsection{Local Differential Privacy}
Local Differential Privacy (LDP) \cite{kasiviswanathan2011can}, a variation of Differential Privacy (DP) \cite{dwork2016calibrating}, eliminates the requirement for a trusted data collector. This makes LDP a highly practical privacy-preserving solution across various applications, such as protecting location privacy during interactions with LBS \cite{xiao2015protecting,meehan2021location}.

\begin{definition}[$\varepsilon$-LDP)]
\label{def:LDP}
    Given a privacy budget $\varepsilon > 0$, a randomized mechanism $\mathcal{M}: \mathcal{X} \rightarrow \mathcal{Y}$ provides $\varepsilon$-LDP if and only if, for any two inputs $x, x' \in \mathcal{X}$ and any possible output $y \in \mathcal{Y}$, the following inequality holds:
\[
\Pr[\mathcal{M}(x) = y] \leq e^{\varepsilon} \times \Pr[\mathcal{M}(x') = y].
\]
\end{definition}
LDP enables users to locally perturb their data in order to guarantee plausible deniability, which is controlled by the privacy budget $\varepsilon$.

\subsection{Adversarial LDP}
In order to construct an effective privacy model, it is crucial to accurately quantify the amount of information that can be inferred by an adversary regarding the real location within publicly available data. The adversary's beliefs undergo revision each time a new location is released. These beliefs are depicted as a probability distribution reflecting the potential user's location among a set of possibilities. In particular, at any timestep $t$, the prior probability of the user's current position $l^*$ among a set of possible position $\mathcal{L}_t$ is denoted by $$Pr^-_{t} [l=l^*] = Pr[l=l^*|z_{t-1}, \cdots, z_1], l \in \mathcal{L}_t$$

Then, to capture the process of belief updating, we can use Bayesian inference to compute the posterior probability after observing the released position $z_t$, which is formally described as follows:
$$
Pr^+_{t}[l=l^*] = \frac{Pr[z_t|l=l^*] \cdot Pr^-_{t} [l=l^*]}{\sum_{l' \in \mathcal{L}_j} Pr[z_t|l'=l^*] \cdot Pr^-_{t} [l'=l^*]}, l \in \mathcal{L}_t
$$

In the adversarial definition of privacy, an adversary is modeled as a probability distribution, representing the adversary's prior belief. An algorithm is considered private if, after observing the algorithm's output, the posterior distribution (the adversary's updated belief) remains close to the prior distribution.

\begin{definition}[Adversarial LDP \cite{rastogi2009relationship}]
\label{def:ALDP}
Given a privacy budget $\varepsilon > 0$, A mechanism $\mathcal{M}:\mathcal{X} \rightarrow \mathcal{Y}$ satisfies $\epsilon$-ALDP if and only if, for any output $y \in \mathcal{Y}$, any input $x \in \mathcal{X}$, and any adversary knowing $\mathcal{X}$, the following holds:
\begin{equation}
    e^{-\epsilon} \cdot Pr[x] \leq Pr[x|y] \leq e^{\epsilon} \cdot Pr[x] 
    \nonumber
\end{equation}
\end{definition}

As shown in \cite{xiao2015protecting}, for continuous location sharing, we can derive 
from the PTLM property \cite{rastogi2009relationship} an equivalence between LDP (Definition \ref{def:LDP}) and ALDP (Definition \ref{def:ALDP}).

\begin{theorem}[\cite{xiao2015protecting}]
For the problem of continual location sharing, Definition \ref{def:ALDP} is equivalent to Definition \ref{def:LDP}.
\end{theorem}

The previous theorem establishes that if a mechanism $\mathcal{M}$ satisfies $\epsilon$-ALDP, then it also satisfies $\epsilon$-LDP. Consequently, for the sake of consistency, we will exclusively refer to $\epsilon$-LDP throughout.

\subsection{Sensitivity Hull}
The concept of Sensitivity Hull \cite{xiao2015protecting} is based on in the extensively researched notion of Convex Hull within computational geometry.

\begin{definition}[Convex Hull]
    Let $X$ be a set of points in a Euclidean space. The convex hull of $X$ is the smallest convex set that contains $X$.
\end{definition}

Since Convex Hull is a well-explored concept, we forgo intricate explanations and simply utilize $CH(X)$ to represent the process of determining the convex hull of $X$. Given our focus on interaction with LBS, where locations are typically represented as 2-dimensional points, it's important to note that $CH(X)$, where $X$ represents a set of locations, forms a polygon.

\begin{definition}[Sensitivity Hull \cite{xiao2015protecting}]
    Given a set of locations $X$, the sensitivity hull of a function $f$ is $CH(\Delta f[X])$ where 
    \begin{equation}
        \nonumber
        \Delta f[X] = \bigcup\limits_{x_1, x_2 \in X} \{f(x_1)-f(x_2)\}
    \end{equation}
\end{definition}

\subsection{Measuring Utility}
\label{sec:utility_error}
In accordance with previous studies such as \cite{xiao2015protecting,shokri2011quantifying}, we evaluate the effectiveness of the obfuscation technique by quantifying its utility. This is done by calculating the Euclidean distance between the actual location $l^{*}$ and the released (obfuscated) location $z$, denoted as:
\begin{equation}
err = d(l^*, z)
\nonumber
\end{equation}
Here, $d$ represents the Euclidean metric.

\section{System Architecture and Adversary Model}
\label{sec:system_arch}
The proposed system architecture in the paper is a third-party-free framework consisting solely of users and LBS servers. In this architecture, users have the ability to autonomously perform location anonymization based on their individual privacy preferences. By obfuscating their personal data on their devices prior to transmission, users can enhance their overall experience with LBS while safeguarding their privacy. Furthermore, the paper focuses on continuous LBS, where users periodically or continuously report their current positions to obtain services. Adversaries in this context possess a collection of multiple or sequential queries gathered over time, creating a strong spatiotemporal correlation. Consequently, attackers can potentially track a specific user, reconstruct their trajectory, and infer the set of visited points of interest.

In our study, we analyze a strong adversary with significant capabilities. We posit that this adversary possesses supplementary information, which may enhance their ability to exploit available information. Specifically, we assume that if only a subset of potential locations is chosen for the release of a perturbed position (i.e., a perturbed position that make the real position (somewhat) indistinguishable within the chosen subset of potential locations), a strong adversary would possess the means to leverage their additional information and accurately determine which specific subset of candidate locations has been taken into consideration.

\section{Warm Up  : A Trivial Model}
\label{sec:trivial_model}
In our approach, we adopt a definition of location obfuscation based on LDP. We choose this framework because, as discussed in Section \ref{sec:related_work}, it provides the means to formally define and analyze location obfuscation mechanisms. Informally, location obfuscation aims to hide the precise location $l$ of an entity by making it indistinguishable within a range of potential locations $\mathcal{L}$, making it challenging to determine the exact one. 

We formalize the LDP-based location obfuscation based on the adversarial LDP (Definition \ref{def:ALDP}).

\begin{definition}[Location-based $\epsilon$-LDP \cite{xiao2015protecting}]
\label{def:adv-location_obfuscation}
Given a set of possible locations $\mathcal{L}$, a mechanism $\mathcal{M}$ provides location-based $\epsilon$-LDP, if and only if, for any location $l \in \mathcal{L}$, any output $z$, and any adversary knowing that the true location $l^*$ is in $\mathcal{L}$, the following holds
\begin{equation}
     Pr[l^* = l | z] \leq e^\epsilon \cdot  Pr[l^* = l]
    \nonumber
\end{equation}
\end{definition}

\begin{remark}
The definition in \ref{def:adv-location_obfuscation} focuses on the upper bound, ensuring that observing the output does not significantly increase the adversary's confidence in identifying the true location (the \emph{presence} property). A more general two-sided bound,
\[
   e^{-\epsilon} \cdot Pr[l^* = l] \leq Pr[l^* = l | z] \leq e^\epsilon \cdot Pr[l^* = l]
   \nonumber
\]
would also include the \emph{absence} property, preventing a significant decrease in this probability. However, since the primary objective is to limit the adversary’s ability to infer the true location, the lower bound is omitted for simplicity.
\end{remark}

\subsection{Capturing Trajectory Privacy}
Location obfuscation can be taken a step further by extending it to obfuscate the trajectory of the user among a set of possible trajectories, i.e., making the real trajectory DP indistinguishable within a set of possible trajectories. To achieve this, for every timestep, denoted as $t$, within the trajectory, and given the set of possible trajectories denoted as $\mathcal{T}$, the key concept is to construct the set of locations, denoted as $\mathcal{L}_t$, based on the potential locations the user could be if it followed each trajectory within $\mathcal{T}$. Specifically, at any given timestep $t$, we suppose that for each possible location $l$ in $\mathcal{L}_t$, and for each trajectory $T$ in $\mathcal{T}$, a probability representing the adversary's belief about the likelihood of reaching location $l$ via trajectory $T$ will be assigned. This probability is denoted as $Pr[l \leadsto T]$.

The following proposition uses the sequential composition theorem of LDP \cite{dwork2016calibrating} to bound the adversarial trajectory privacy when using an LDP location obfuscation mechanism.

\begin{proposition}
\label{prop:trajectory_privacy}
Given a set of trajectories $\mathcal{T}$. Let us denote by $\mathcal{L}_t$ the set of possible  positions inside each trajectory in $\mathcal{T}$ during the position release timestep $t \in [1,n]$. If we use an $\epsilon_i$-LDP location mechanism $\mathcal{M}_i$ for releasing the position a timestep $i$, then for any release timestep $t$, any output $z_i$, and any $T \in \mathcal{T}$, the following condition holds:
\begin{equation}
    Pr[T^* = T] \leq e^{\sum_{i=1}^t \epsilon_i} \cdot Pr[T^* = T|z_1, \cdots, z_t] \nonumber
\end{equation}
\end{proposition}

The proof of the preceding proposition is provided in Appendix \ref{proof:prop:trajectory_privacy}.

\subsection{Limitation}
In several scenarios, multiple interactions with the LBS occur at different timestep. In such cases, it becomes necessary not only to ensure the privacy of individual locations but also to consider protecting the user's trajectory. In this section, we explore the limitations of using DP for location obfuscation when aiming to preserve the privacy of the user's trajectory and visited POIs.

Therefore, we employ the LDP based location obfuscation technique and demonstrate that it is challenging to achieve a satisfactory balance between trajectory privacy and the utility derived from the LBS.

\begin{figure}[h!]
    \centering
    \begin{subfigure}[b]{0.22\textwidth}
        \centering
        \includegraphics[width=\textwidth]{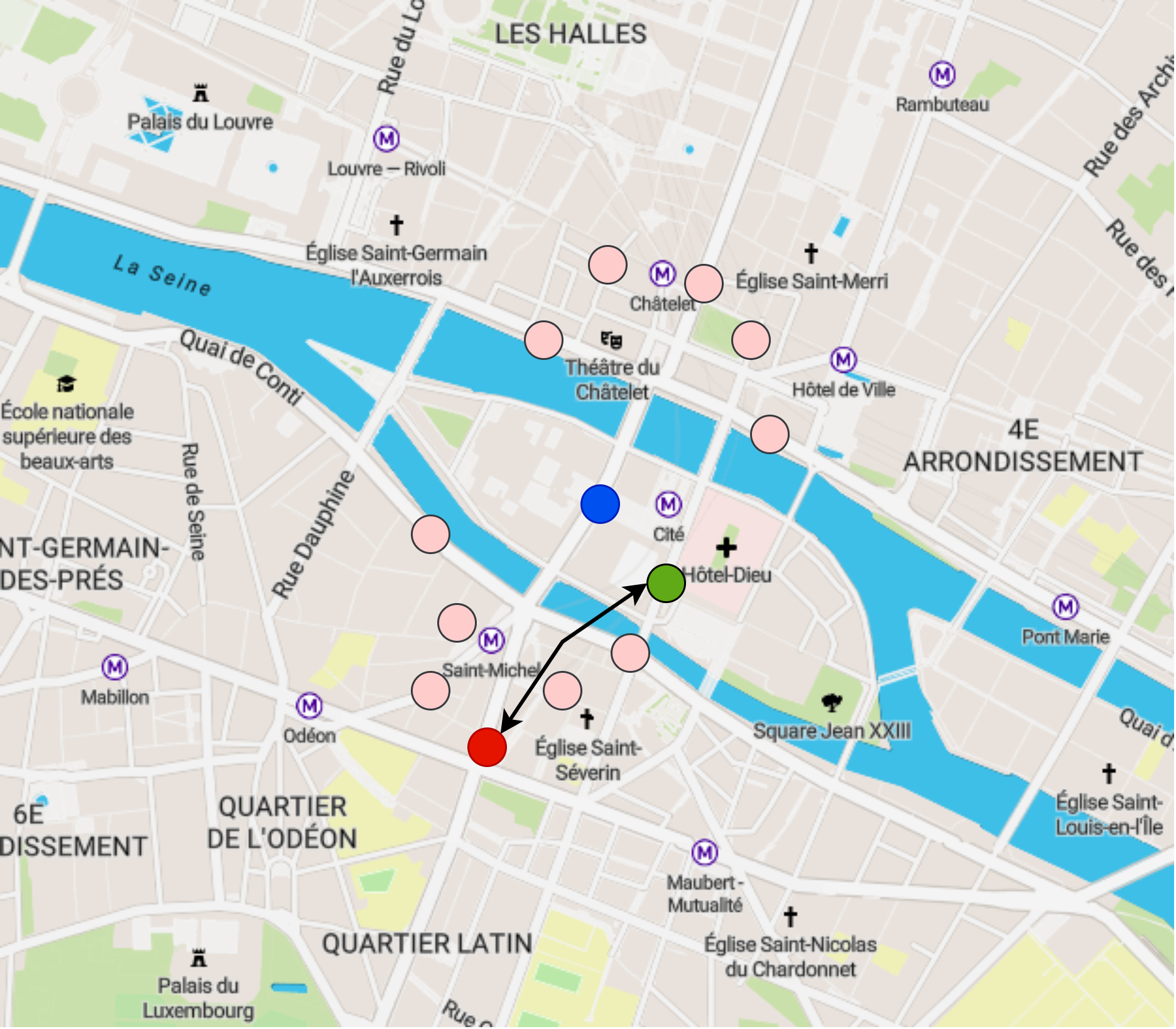}
        \caption[]%
        {{\small Timestep 1: \texttt{err}$ = 0.2$ \texttt{km}}}    
        \label{fig:mean and std of net14}
    \end{subfigure}
    \hfill
    \begin{subfigure}[b]{0.22\textwidth}  
        \centering 
        \includegraphics[width=\textwidth]{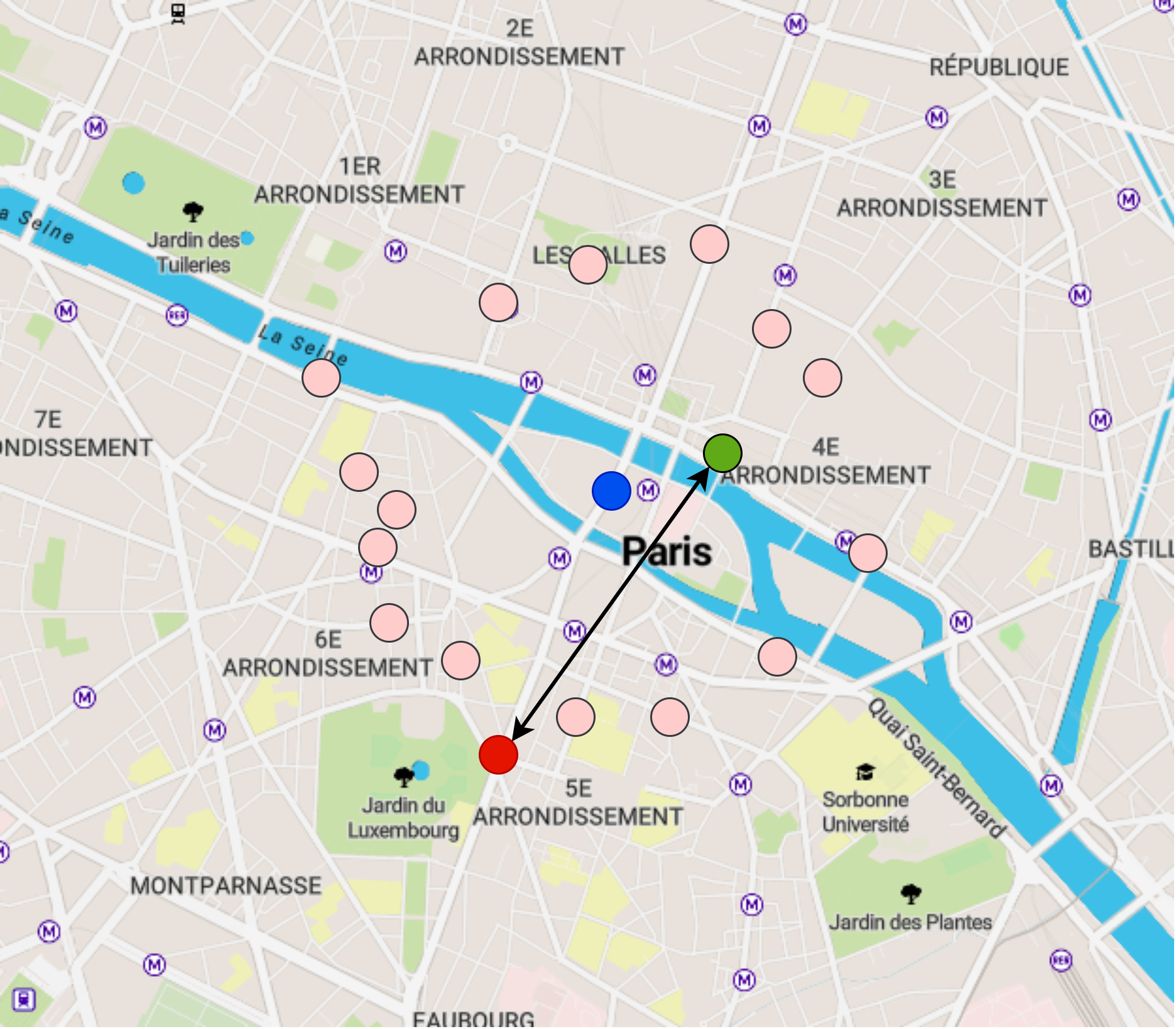}
        \caption[]%
        {{\small Timestep 2: \texttt{err}$ = 1.2$ \texttt{km}}}    
        \label{fig:mean and std of net24}
    \end{subfigure}
    \vskip\baselineskip
    \begin{subfigure}[b]{0.22\textwidth}   
        \centering 
        \includegraphics[width=\textwidth]{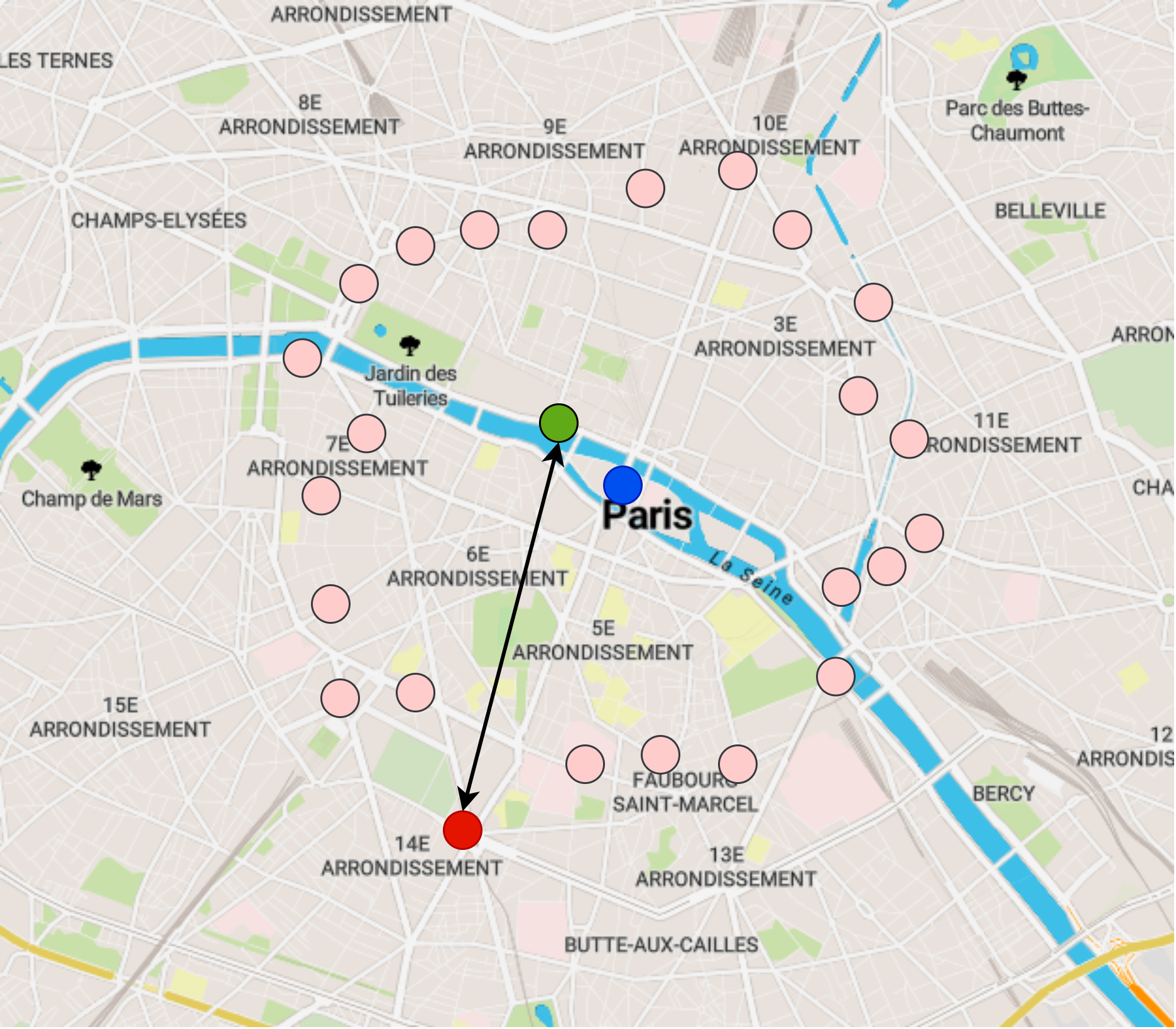}
        \caption[]%
        {{\small Timestep 3: \texttt{err}$ = 2.7$ \texttt{km}}}    
        \label{fig:mean and std of net34}
    \end{subfigure}
    \hfill
    \begin{subfigure}[b]{0.22\textwidth}   
        \centering 
        \includegraphics[width=\textwidth]{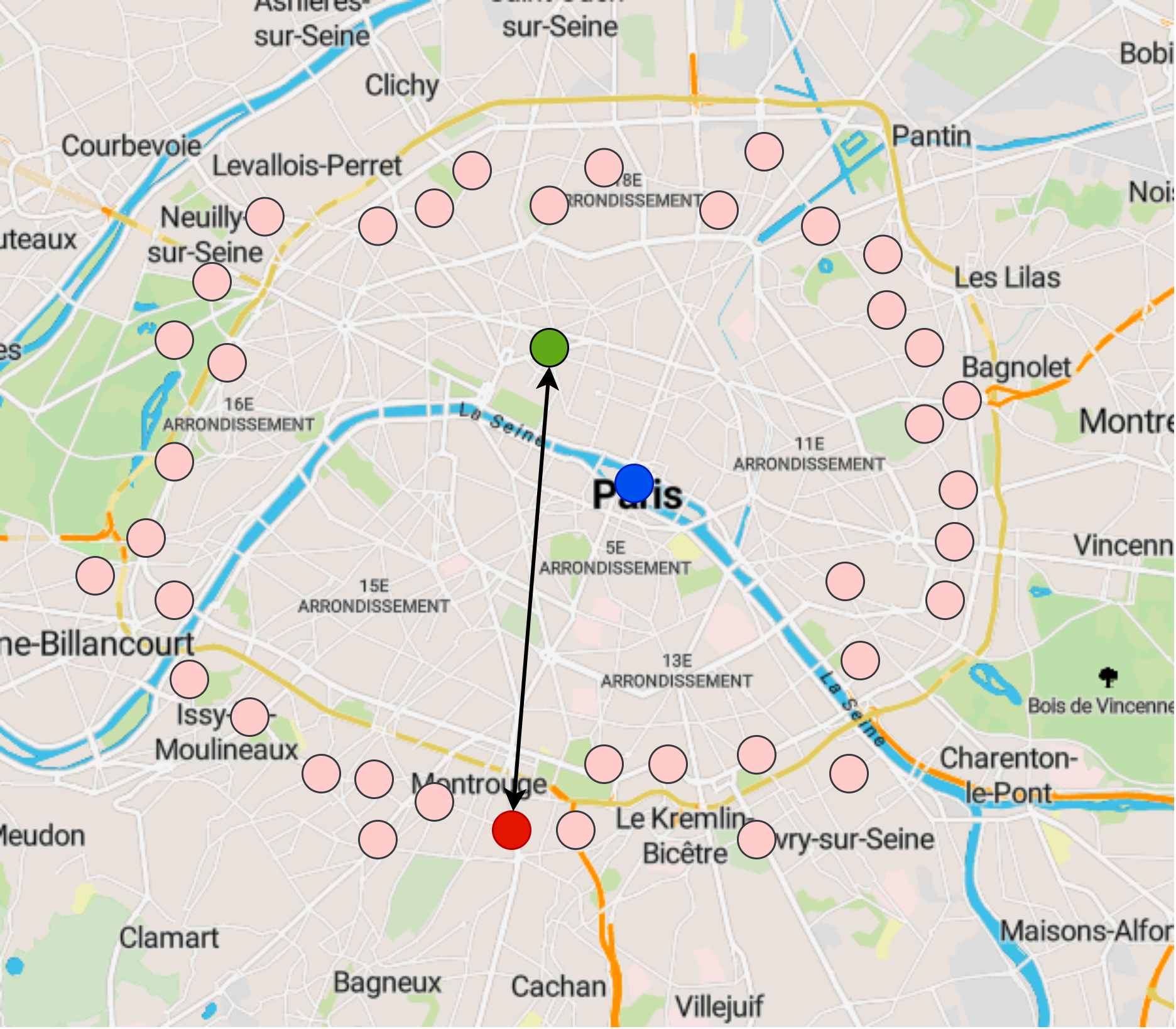}
        \caption[]%
        {{\small Timestep 4: \texttt{err}$ = 6.2$ \texttt{km}}}    
        \label{fig:mean and std of net44}
    \end{subfigure}
    \caption[]
    {\small Temporal Evolution of captured utility with LDP location obfuscation mechanism. The initial position of the user is denoted by the blue point ({\color{blue}$\bullet$}) in the figure. The pink points ({\color{babypink}$\bullet$}) represent the adversary's perception of the potential positions at a specific timestep, which is inferred based on various background information available to the adversary, such as road traffic conditions. The red point ({\color{red}$\bullet$}) indicates the actual (true) position of the user, while the green point ({\color{green}$\bullet$}) represents the obfuscated position. Finally, the arrows $\leftrightarrow$ depicts \texttt{err}, the loss of utility quantified as the distance between the obfuscated shared location and the actual location of the user. In each on the timesteps, obfuscated position was derived using the
Planar Isotropic Mechanism \cite{xiao2015protecting}.}  
    \label{fig:mean and std of nets}
\end{figure}

Figure 1 illustrates the temporal evolution of captured utility over four consecutive timesteps. The experiment involves the implementation of location-based obfuscation techniques \cite{xiao2015protecting} to protect the user's trajectory. It is worth mentioning that the released position was obfuscated using the differentially private Planar Isotropic Mechanism \cite{xiao2015protecting}. The reported figures clearly illustrate that the utility provided by the location obfuscation technique diminishes rapidly as more perturbed locations are released. This decline is primarily attributed to our consistent effort to make the released location LDP-indistinguishable to all potential positions for each timestep, effectively concealing the user's trajectory among all possible paths. By striving to ensure the previous, the technique's utility is compromised, resulting in a swift decrease in its provided utility.

\subsection{Towards An Adversary-Aware Privacy Framework}
The limitations of the previous, the aforementioned trivial approach highlight that indiscriminately obfuscating all positions—irrespective of the adversary’s potential knowledge—quickly erodes utility, as the masked locations become increasingly imprecise over time. A radius-based method, such as \cite{andres2013geo,meehan2021location}, may achieve LDP within a defined geographic region, but it fails to impose a consistent privacy bound across all possible locations. By factoring in the adversary’s prior knowledge, our framework is able to accurately quantify privacy loss for a given adversary profile, optimizing utility while maintaining robust privacy guarantees. Although this results in privacy bounds that depend on specific adversarial assumptions, we believe this targeted approach marks a significant advancement, providing both practical and measurable privacy protections.

\section{The Proposed Model}
\label{sec:the_proposed_model}
In this section, we define a model that allows a flexible compromise between the utility captured through the usage of the LBS and the level of ensured privacy.
We first formulate a model that can bound the privacy loss arising from the disclosure of perturbed positional data, considering the location perspective. We then extend the proposed model to encompass the privacy loss, considering both the trajectory and the visited points of interest perspectives.

\subsection{Intuition}
To overcome the limitation of the trivial model proposed in Section \ref{sec:trivial_model}, the entity responsible for obfuscation, i.e., the user, must strike a balance between the level of utility provided (i.e., the distance between the released perturbed position and the actual one) and the level of ensured privacy. To achieve this, during each interaction with the LBS, the user must adaptively obfuscate the real position of the user using a subset of candidate locations, denoted by $\mathcal{L}^*$, out of all possible locations $\mathcal{L}$ to achieve a specific utility/privacy trade-off.

\subsection{Privacy Leakage}
Building on the previous intuition, the leakage profiles, illustrated in Figure \ref{fig:leakage-profile}, can be categorized into two main types:
\begin{itemize}
\item Global leakage: It captures the change in an adversary's advantage in knowing the real position of the user before and after the selection of the set of locations $\mathcal{L}^*$ used to perturb the actual user location. In particular, we consider powerful adversary that can potentially deduce the set of locations that were not considered by the obfuscation mechanism $\mathcal{L} \backslash \mathcal{L}^*$. This knowledge can significantly enhance the adversary's ability to identify the true position of the user, and thus needs to be considered in quantifying privacy leakage.

\item Local leakage: it refers to the assessment of how the release mechanism affects an adversary's ability to determine the true position of the user within a specific subset of candidate locations $\mathcal{L}^*$. It involves quantifying the information revealed about the actual user location before and after the perturbed location is released, given the known subset of candidate locations (i.e., the global leakage).
\end{itemize}

\begin{figure}[H]
    \centering
    \includegraphics[scale=0.09]{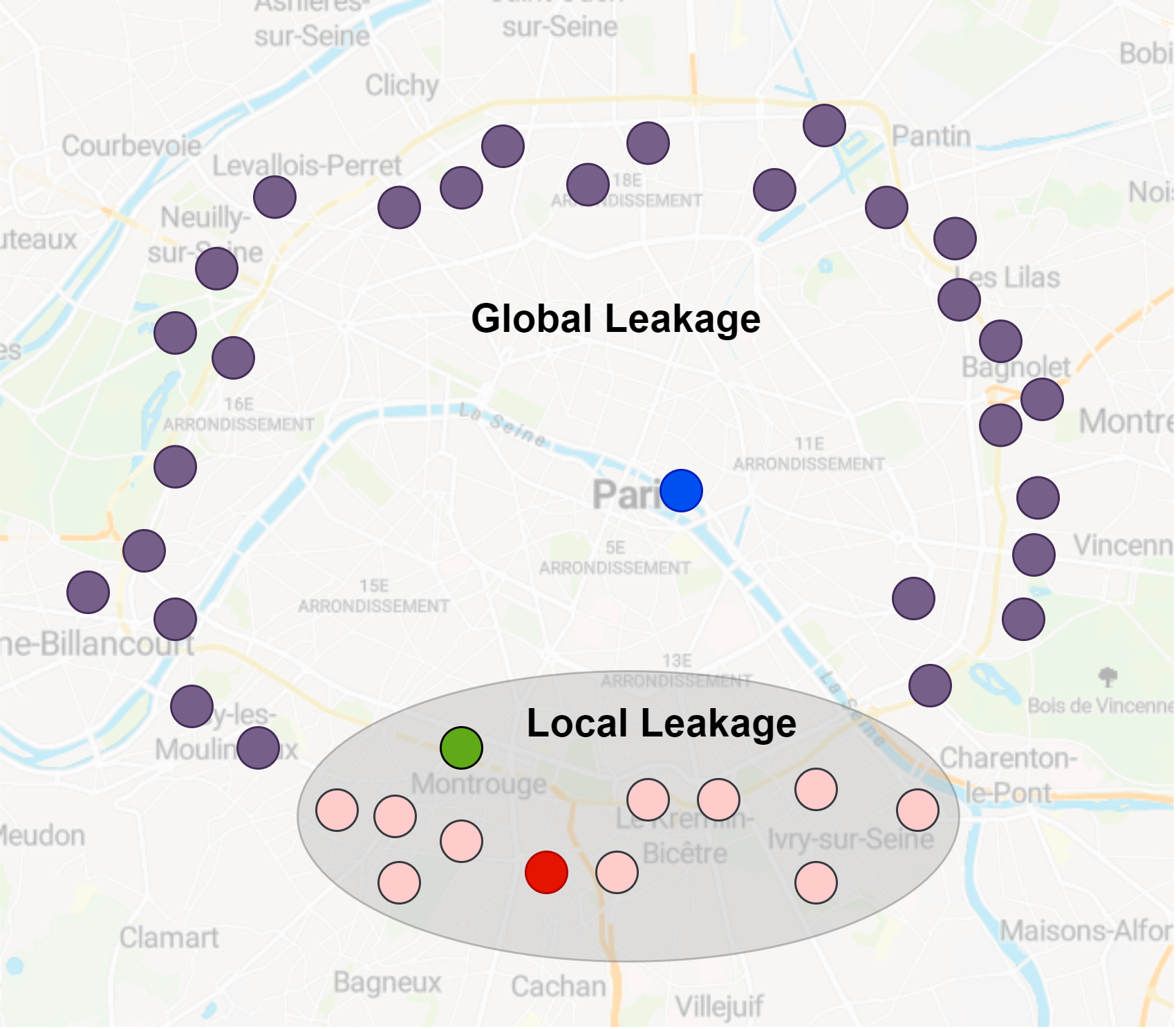}
    \caption{The leakage profiles considered in our model. The blue point ({\color{blue}$\bullet$}) marks the initial position of the user. In contrast, the pink points ({\color{babypink}$\bullet$}) represent a subset of candidate locations chosen for obfuscating the actual user's location, which is symbolized by the red point ({\color{red}$\bullet$}). Lastly, the violet point ({\color{violet}$\bullet$}) signifies potential locations that were excluded from the process of obfuscating the current user's location.}
    \label{fig:leakage-profile}
\end{figure}

\subsection{Location-based Privacy Loss}
Similarly to the approach proposed in \cite{xiao2015protecting}, at a timestep $t$, we define $\delta$-obfuscation set to be the set of candidate locations $\mathcal{L}^*$ to be used to obfuscate the real location of the user.

\begin{definition}[$\delta$-obfuscation set] 
\label{def:delta_los}
Let $\mathcal{L}_t$ be set of possible position at timestep $t$, and $Pr^-_t$ be the prior probability of the user current location $l_t^*$. At the timestep $t$, the $\delta$-obfuscation set , denoted by $\mathcal{L}^*_t$, is a subset of locations $\mathcal{L}_t$ such that the
prior probability sum of $l_i = l^*, l_i \in \mathcal{L}^*_t$ equals to $\delta$.
\begin{equation}
    \nonumber
    \mathcal{L}^*_t = \left\{l_i \in \mathcal{L}|\sum_{l_i}Pr^-_t[l_i = l^*_t] = \delta \right\}
\end{equation}
\end{definition}

\begin{remark}
    In the sequel, for the sake of simplicity, we will assume that the true location of the user at timestep $t$ belongs to the selected $\delta$-obfuscation set $\mathcal{L}^*_t$. Please note that when the previous assumption is invalid, we can employ the surrogate location, i.e., the location in $\mathcal{L}^*_t$ with the shortest distance to the true location of the user, as if it were the true location in the release mechanism, as proposed in \cite{xiao2015protecting}.
\end{remark}

The concept of the $\delta$-obfuscation set, combined with the ALDP concept (Definition  \ref{def:adv-location_obfuscation}
), allows us to encompass the principle of LDP on $\delta$-obfuscation set, with the intuition that, for a timestep $t$, the usage of an $\epsilon$-LDP mechanism to released location $z_t$ will not help an adversary to link the real location of user and any location in the $\delta$-obfuscation set  by more than an $\epsilon$-fold factor.

\begin{definition}[LDP on $\delta$-obfuscation set]
\label{def:DP_los}
Given a timestep $t$ and a $\delta$-obfuscation set $\mathcal{L}^*_t$, a mechanism $\mathcal{M}$ satisfied $\epsilon$-LDP on a set of locations $\mathcal{L}^*_t$ if and only if, for any output $z_t$, and any locations $l_i$ in $\mathcal{L}^*_t$, and any adversaries knowing the true location $l^*_t$ is in
$\mathcal{L}^*_t$, the following holds:
\begin{equation}
    \nonumber
    Pr_t^+[l_i = l_t^*] \leq e^\epsilon \cdot Pr_t^-[l_i = l^*_t]
\end{equation}
\end{definition}

\begin{remark}
    According to the previous definition, a differentially private mechanism $\mathcal{M}$ ensures $\epsilon$-LDP of the true location $l^*_t$ within $\mathcal{L}^*_t$. We observe that as the value of $\delta$ approaches $1$, the privacy level it ensures becomes stronger, as the real location of the user will be $\epsilon$-LDP from the most probable set of locations.
\end{remark}

Above definition allows to bound the local leakage on the real position of the user according to the set of considered candidate positions in the $\delta$-obfuscation set. Nevertheless, it does not provide a means to measure the overall information leakage on the location of the user that arises from the selection of candidate locations included in the  $\delta$-obfuscation set (i.e., the global leakage) and the released location. Given our assumption of a strong adversary capable of accurately deducing the $\delta$-obfuscation set under consideration at each time step, this knowledge can be leveraged to enhance their confidence in determining the actual location of the user.

The subsequent theorem introduces the concept of $(\epsilon,\delta)$-location privacy, which permits the establishment of an upper limit on the overall privacy loss incurred when employing a LDP location mechanism over a $\delta$-location obfuscation set to disclose a perturbed location. The parameter $\delta$ bounds the global privacy loss that arises from the selection of a subset of locations, specifically denoted by $\delta$-obfuscation set used to hide the true user location. On the other hand, the parameter $\epsilon$ is employed to constrain the local privacy loss i.e., local leakage within the $\delta$-obfuscation set, incurred when disclosing the perturbed location data.

\begin{definition}[$(\epsilon,\delta)$-location privacy] 
    Let $\mathcal{L}_t$ be the set of possible positions at timestep $t$, $\mathcal{L}^*_t$ be the $\delta$- obfuscation set to be used at timestep $t$. A mechanism $\mathcal{M}$ satisfies $(\epsilon,\delta)$-location privacy on $\mathcal{L}$ if and only if, for any output $z_t$, any locations $l_i$ in $\mathcal{L}$, and any adversaries knowing the true location $l^*_t$ is in $\mathcal{L}^*_t$, the following conditions hold:
    \begin{numcases}{}
        Pr^+[l_i = l^*_t] = 0, & if $l_t \notin \mathcal{L}^*_t$ \label{eq:thm1_first} \\
        Pr^+[l_i = l^*_t | l^*_t \in \mathcal{L}^*_t] \leq \delta^{-1} \cdot e^\epsilon \cdot Pr^-[l_i = l^*_t],              & otherwise \label{eq:thm1_second}
    \end{numcases}
\end{definition}

In the previous definition, Equation \ref{eq:thm1_first} states that upon defining the $\delta$- obfuscation set, labeled as $\mathcal{L}^*_t$, and assuming that the user's actual location, denoted as $l^*$, must be within $\mathcal{L}^*_t$, the posterior probability (the adversary belief) of any other potential location, which fall outside of $\mathcal{L}^*_t$, being the correct user location $l^*$ is zero. On the other hand, Equation \ref{eq:thm1_second} defines the upper limit on privacy loss by establishing a connection between the initial probability that the candidate position $l_i$ is the actual user's position and the posterior probability that arises once we have defined the $\delta$-location obfuscation set and released the perturbed position $z$. 

\begin{theorem}
\label{thm:location-privacy}
    Let $\mathcal{L}_t$ be the set of possible positions at timestep $t$, $\mathcal{L}^*_t$ be the $\delta$- obfuscation set to be used at timestep $t$. If a mechanism $\mathcal{M}$ satisfies $\epsilon$-LDP on $\mathcal{L}^*_t$, then, for any output $z_t$, any locations $l_t$ in $\mathcal{L}$, it satisfies $(\epsilon,\delta)$-location privacy on $\mathcal{L}$.
\end{theorem}

The preceding theorem asserts that in the presence of a strong adversary capable of accurately deducing the considered $\delta$-obfuscation set out of all probable locations $\mathcal{L}$ and knowing that $l^*_t \in \mathcal{L}_t^*$, if mechanism $\mathcal{M}$ satisfies $\epsilon$-LDP on $\mathcal{L}^*_t$, it inherently satisfies $(\epsilon,\delta)$-location privacy on $\mathcal{L}_t$. The
proof of the previous theorem is provided in Appendix \ref{proof:theorem_location_privacy}.

\subsection{Trajectory-based Privacy Loss Quantification}
The previous section bounds the privacy loss concerning the user location when both the set of positions employed to obfuscate the actual user location and the released perturbed location become accessible to an adversary. In the subsequent, we broaden the model's scope to encompass the loss of privacy in relation to the trajectory of the user, under the same conditions where the adversary possesses knowledge of both the positions used for obfuscation and the resulting perturbed location.

To achieve the aforementioned objective, we introduce the notion of potential trajectories, which we represent as $\mathcal{T}$. In practical terms, these trajectories can be constructed using a combination of the possible destinations  within the designated geographic area and the user's initial location. The set of trajectories, denoted as $\mathcal{T}$, encompasses the shortest paths, whether measured in time or distance, from the user's current starting position to the considered destinations.

We now have all the necessary elements to establish a privacy bound concerning trajectories. The following definition introduces ($\epsilon, \delta, \theta$)-trajectory privacy, which allows to bound the privacy loss with respect to trajectories when both the set of positions used to obfuscate the true user location and the disclosed perturbed location are accessible to an adversary.

\begin{definition}[($\epsilon, \delta, \theta$)-trajectory privacy]
\label{def:trajecroty_privacy}
    Let $\mathcal{T}$ be the set of possible trajectories, $T^*$ be the true trajectory of the user, $\mathcal{L}^*_t$ be the $\delta$-obfuscation set to be used at timestep $t$, a mechanism $\mathcal{M}$ satisfies ($\epsilon, \delta, \theta$)-trajectory privacy if and only if, for all trajectory $T \in \mathcal{T}$, for any output $z_t$ representing the released location, and any adversaries knowing the true location $l^*_t$ is in $\mathcal{L}^*_t$, the following holds:
    \begin{equation}
        Pr^+[T = T^* | l^* \in \mathcal{L}^*] \leq \delta^{-1} \cdot e^\epsilon \cdot Pr^-[T = T^*] + \theta
        \nonumber
    \end{equation}
   
\end{definition}

We can now formally establish a connection between ($\epsilon, \delta$)-location privacy and ($\epsilon, \delta, \theta$)-trajectory privacy through the following theorem, where we demonstrate that if a mechanism $\mathcal{M}$ satisfies ($\epsilon, \delta$)-location privacy, then it also satisfies ($\epsilon, \delta, \theta$)-trajectory privacy for a specific value of $\theta$.

\begin{theorem}
    \label{thm:link-LP-TP}
    Let $\mathcal{T}$ be the set of possible trajectories, $\mathcal{L}_t$ be the set of possible positions at timestep $t$, and $\mathcal{L}^*_t$ be the $\delta$-obfuscation set to be used at timestep $t$. A mechanism $\mathcal{M}$ that satisfies $(\epsilon, \delta)$-location privacy for any output $z_t$ representing the released location and any adversaries knowing the true location $l^*_t$ is in $\mathcal{L}^*_t$, satisfies 
    ($\epsilon, \delta, \theta$)-trajectory privacy for any trajectory $T \in \mathcal{T}$
    where 
    \begin{equation}
    \theta = (\delta-1) \cdot \min\limits_{l \in \mathcal{L}\backslash \mathcal{L}^*} Pr[l \leadsto T]
    \nonumber
    \end{equation}
\end{theorem}

The preceding theorem establishes the upper bound on privacy leakage for any trajectory within $\mathcal{T}$ when a (perturbed) location is disclosed through a $(\epsilon, \delta)$-location privacy mechanism. The proof of the preceding theorem is provided in Appendix \ref{proof:link-LP-TP}.

\begin{comment}
    \begin{remark}
    It is important to highlight that our model characterizes the trajectory as a path connecting the user's initial position to a particular point of interest. Consequently, extending our model to assess privacy loss with respect to individual points of interest is straightforward.
    \end{remark}
\end{comment}

In the course of a trajectory, a user may interact with an LBS across multiple timesteps, making it imperative to comprehend the aggregate impact of these interactions on user's trajectory. Understanding how to delineate upper bounds on trajectory privacy loss, especially when divulging multiple locations over distinct time intervals, is crucial. In the following, we present a theorem that elucidates the precise interplay of the parameters $\epsilon$, $\delta$, and $\theta$ when multiple location privacy preserving mechanisms are employed to release perturbed locations at different timesteps.

\begin{theorem}
    \label{thm:composition}
    Let $\mathcal{M}_t$, $t \in [1,n]$, be an ($\epsilon_t, \delta_t, \theta_t$)-trajectory private mechanism. If we define $\mathcal{M}_{[n]} = \{\mathcal{M}_1, \mathcal{M}_2, \cdots, \mathcal{M}_n\}$ to be the mechanism that uses $\mathcal{M}_t$ to release a perturbed location $z_t$ during the timestep $t$, then $\mathcal{M}_{[n]}$ is ($\epsilon_{[n]}, \delta_{[n]}, \theta_{[n]}$)-trajectory private where:
    \begin{equation}
        \begin{split}
            \epsilon_{[n]} & =  \sum_{i=1}^n \epsilon_i \\
            \delta_{[n]} & =  \prod_{i=1}^n \delta_i^{-1}
            \\
            \theta_{[n]} & =  \theta_n + \sum_{i=1}^{n-1} \theta_i \prod_{j=i+1}^{n} \delta_{j}^{-1} \cdot e^{\epsilon_{j}}
        \end{split}
        \nonumber
    \end{equation}
\end{theorem}

The importance of Theorem \ref{thm:composition} lies in its acknowledgment of the diverse parameters ($\epsilon_t, \delta_t, \theta_t$) associated with different mechanisms $\mathcal{M}_t$ utilized across a trajectory. This recognition forms the foundation for constructing a flexible framework enabling users to tailor their privacy preferences for each interaction with an LBS. By delineating how these parameters interact cumulatively, the theorem facilitates the creation of personalized privacy measures, empowering users to assess and customize their privacy loss dynamically throughout their interactions with LBSs. The proof of the previous theorem is provided in Appendix \ref{proof:composition}.

\subsection{POI-based Privacy Loss Quantification}
In addition to the revealed information about the user's current location and trajectory, interacting with location navigation services can inadvertently expose sensitive details about the visited points of interest. To illustrate, the us suppose a geographic area in which a set of POIs $\{P_1, P_2, P_3\}$ exists. Suppose that for each $P_i$, we denote by $P_i^t$ the average time an individual spends at $P$, and that $P_1^t = 30, P_2^t = 7, P_3^t = 8$. Consider a scenario where a user engages with LBS, revealing two locations $l_t$ and $l_{t+1}$ during two consecutive time steps $t_j$ and $t_{j+1}$, respectively. The time elapsed between these steps is denoted as $t_{j+1} - t_j = 34$ minutes. Both locations, $l_t$ and $l_{t+1}$, fall within the same geographic area containing a set of POIs $\mathcal{P}$. Suppose, based on traffic information, the anticipated travel time between $l_t$ and $l_{t+1}$ is approximately 7 minutes. In this scenario, an adversary, armed with access to the disclosed locations and traffic details, could reasonably deduce that the user likely visited the Point of Interest $P_1$.

In this section, we broaden the model's scope to encompass the assessment of privacy loss associated with the POIs visited by the user, in identical hypothesis, where the adversary is acquainted with both the positions used to obfuscate the real location of the user and the disclosed perturbed location. To achieve the aforementioned objective, we introduce the concept of possible points of interest, defined as the set of POIs located within the designated geographic area denoted by $\mathcal{P}$. Subsequently, we posit that each point of interest $P_i \in \mathcal{P}$ is associated with a probability of being visited by the user between two locations released in two consecutive time steps $t_i$ and $t_{i+1}$ denoted by $Pr_{t+1}[P^* = P]$, were $P^*$ is the real visited POI. These probabilities encapsulate the adversary's belief regarding the likelihood of a POI being visited by the user during the interval between two consecutive timesteps. Specifically, drawing inspiration from the aforementioned example, this probability can be simply computed as a nuanced function. It takes into account factors such as the average time spent by individuals in the POI and the time required by the user to traverse between the two disclosed locations, factoring in real-time road traffic conditions for a more precise estimation. To enhance the precision of the probability computation, more sophisticated information can be integrated. Studies such as those by He et al. \cite{he2016inferring} on behavior patterns, category-aware considerations as demonstrated by He, Li, and Liao \cite{he2017category}, and the incorporation of various contexts, as explored by Yang et al. \cite{yang2017neural}, have showcased the effectiveness of neural network-based methods in the prediction of location and visited POIs. 

At any given timestep $t$, let $l_i \in \mathcal{L}_t$ represent a possible location and $P_i \in \mathcal{P}$ denote a POI. We define $Pr^-t[l_i \leadsto P_i]$ as the probability that the user visited $P_i$ given the disclosed location $z_{t-1}$ at timestep $t-1$ and $l_i$ as the potential location of the user at timestep $t$.

In the following, we establish a privacy boundary in relation to POI. We introduce the concept of ($\epsilon$, $\delta$, $\theta$)-POI privacy, offering a means to quantify and limit privacy loss in connection with a designated set of POIs.

\begin{definition}[($\epsilon$, $\delta$, $\theta$)-POI privacy]
\label{def:POI_privacy}
    Let $\mathcal{P}$ be the set of possible POIs, $P^*$ be the POI visited by the user at timestep $t$ (between the disclosed location $z_{t-1}$ and $z_t$), $\mathcal{L}^*_t$ be the $\delta$-location obfuscation set to be used at timestep $t$, a mechanism $\mathcal{M}$ satisfies ($\epsilon, \delta, \theta$)-POI privacy if and only if, for any POI $P \in \mathcal{P}$, for any output $z_t$ representing the released location, and any adversaries knowing the true location $l^*_t$ is in $\mathcal{L}^*_t$ the following holds:
    \begin{equation}
        Pr^+_t[P = P^* | l^*_t \in \mathcal{L}^*_t] \leq \delta^{-1} \cdot e^\epsilon \cdot Pr^-_t[P = P^*]  + \theta
        \nonumber
    \end{equation}
\end{definition}

The previous definition aims to formalize the privacy guarantees provided by the mechanism $\mathcal{M}$ regarding the disclosure of the user's visited POIs. It ensures that the mechanism's output does not significantly increase the likelihood of an adversary inferring the true POI $P^*$ beyond what would be expected from prior knowledge, considering the obfuscated location and the specified privacy parameters $\epsilon$, $\delta$, and $\theta$. To ascertain the values of the aforementioned parameters, we introduce a formal theorem that establishes a direct connection between the concept of $(\epsilon, \delta)$-location privacy, where users can delineate a set of locations for conducting LDP-based obfuscation of their real locations, and the concept of $(\epsilon, \delta, \theta)$-POI privacy.

\begin{theorem}
\label{thm:link-LP-POIP}
Let $\mathcal{P}$ denote the possible POIs, $\mathcal{L}_t$ be the set of possible positions at timestep $t$, and $\mathcal{L}^*_t$ be the $\delta$-location obfuscation set to be used at timestep $t$. A mechanism $\mathcal{M}$ that satisfies $(\epsilon, \delta)$-location privacy for any output $z_t$ representing the released location and any adversaries knowing the true location $l^*_t$ is in $\mathcal{L}^*_t$, satisfies ($\epsilon, \delta, \theta$)-POI privacy for any POI $P \in \mathcal{P}$ where 
    \begin{equation}
    \theta = (\delta-1) \cdot \min\limits_{l \in \mathcal{L}\backslash \mathcal{L}^*} Pr[l \leadsto P]
    \nonumber
    \end{equation}    
\end{theorem}

The Proof of the previous theorem is provided in Appendix \ref{proof:link-LP-POIP}.

\subsection{Accuracy Quantification}
In earlier sections, we present our unified framework allowing to quantify simultaneously the loss of privacy concerning user locations, trajectories, and points of interest (POI). These quantifiable losses are primarily influenced by the $\delta$-location set utilized for applying differential privacy-based obfuscation to the actual user location, along with the privacy parameter $\epsilon$. To enable the entity engaging with the LBS to balance the trade-off between privacy loss and captured accuracy, it's imperative to gauge the captured accuracy based on the aforementioned parameters of our framework, namely $\delta$ and $\epsilon$.

To fulfill the previous objective, we provide through the following theorem a formal link between the notion of $(\epsilon,\delta)$-location privacy and the error bound introduced by the mechanisms enforcing $(\epsilon,\delta)$-location privacy when releasing a perturbed position. We note that the error is defined as the distance between the real position and the released position (Section \ref{sec:utility_error}).

\begin{theorem}
\label{thm:utility}
    Given a function d implementing the euclidean metric. Let $\mathcal{L}$ be the possible position of the user and $\mathcal{L}^*$ be the $\delta$-location obfuscation to be used to obfuscate the real location $l^*$ of the user. A mechanism $\mathcal{M}$ ensuring $(\epsilon, \delta)$-location privacy on $\mathcal{L}$ introduces an error $\mathcal{E}$ where
    \begin{equation}
        \label{eq:error}
        \mathcal{E} \geq \Omega \left(\frac{1}{\epsilon}\sqrt{Area(CH(\Delta d[\mathcal{L}^* \cup \{l^*\}]))}\right)
    \end{equation}
    where $CH(\Delta d[\mathcal{L}^* \cup \{l^*\}])$ is the sensitivity hull of $d$ with respect to the set of locations $\mathcal{L}^* \cup \{l^*\}$ and Area is the area of the aforementioned sensitivity hull.
\end{theorem}

\begin{proof}
Since $\mathcal{M}$ satisfies $(\epsilon, \delta)$-location privacy on $\mathcal{L}$, then based on Theorem \ref{thm:location-privacy}, it satisfies $\epsilon$-LDP on $\mathcal{L}^*$. In addition, it has been shown in \cite{xiao2015protecting} (Theorem 4.3) that a mechanism that satisfies $\epsilon$-LDP on a $\delta$-location obfuscation set $\mathcal{L}^*$ has a lower bound error of $\frac{1}{\epsilon} \sqrt{\text{Area}(CH(\Delta d[\mathcal{L}^*]))}$ while supposing that $l^* \in \mathcal{L}^*$, which concludes the proof.
\end{proof}

\section{Achieving $(\epsilon,\delta)$-Location Privacy}
\label{sec:instantiation}
In the previous section, we introduced a comprehensive framework designed to quantify privacy loss concerning locations, trajectories, and Points of Interest (POIs), while also assessing the accuracy of a differential privacy-based mechanism. To bridge theory with practical application, this section aims to instantiate the aforementioned mechanism within our framework, demonstrating its effectiveness in achieving the defined privacy notions.

The following theorem states that releasing perturbed positions using Planar Isotropic Mechanism \cite{xiao2015protecting} satisfies $(\epsilon,\delta)$-location privacy, $(\epsilon,\delta,\theta_t)$-trajectory privacy, and $(\epsilon,\delta,\theta_p)$-POI privacy while simultaneously achieving optimal accuracy.

\begin{theorem}
    Let $\mathcal{T}$ be the set of possible trajectories, $T^*$ denote the true trajectory of the user, $\mathcal{L}_t$ denote the set of possible position at timestep $t$, $\mathcal{L}^*_t$ be the $\delta$-location obfuscation set to be used at timestep $t$, $\mathcal{P}$ denote the set of possible POIs, $P^*$ be the POI visited by the user at timestep $t$. The Planar Isotropic Mechanism satisfies the following:
    \begin{enumerate}
        \item \label{thm:stat1} It ensures $(\epsilon,\delta)$-location privacy on $\mathcal{L}_t$, $(\epsilon,\delta,\theta_t)$-trajectory privacy on $\mathcal{T}$, and $(\epsilon,\delta,\theta_p)$-POI privacy on $\mathcal{P}$ where
        \begin{equation}
        \begin{split}
        \nonumber
            \theta_t &= (\delta - 1) \cdot \min\limits_{l \in \mathcal{L}\backslash \mathcal{L}^*} Pr[l \leadsto T] \\
            \theta_p &= (\delta - 1) \cdot \min\limits_{l \in \mathcal{L}\backslash \mathcal{L}^*} Pr[l \leadsto P] \\
        \end{split}
        \end{equation}
        \item \label{thm:stat2} It provides optimal accuracy with respect to used $\delta$-location obfuscation set $\mathcal{L}^*_t$.
    \end{enumerate}
\end{theorem}

\begin{proof}
    We begin by establishing the validity of Statement \ref{thm:stat1}. According to Xiao et al. \cite{xiao2015protecting} (Theorem 5.1), the Planar Isotropic Mechanism ensures $\epsilon$-LDP within the $\delta$-location obfuscation set $\mathcal{L}^*$. Leveraging this result alongside Theorems \ref{thm:location-privacy}, \ref{thm:link-LP-TP}, and \ref{thm:link-LP-POIP}, we can demonstrate that the Planar Isotropic Mechanism guarantees $(\epsilon,\delta)$-location privacy on $\mathcal{L}$, $(\epsilon,\delta,\theta_t)$-trajectory privacy on $\mathcal{T}$ (where $\theta_t = (\delta-1) \cdot \min\limits_{l \in \mathcal{L}\backslash \mathcal{L}^*} Pr[l \leadsto T]$), and $(\epsilon,\delta,\theta_p)$-POI privacy on $\mathcal{P}$ (with $\theta_p = (\delta-1) \cdot \min\limits_{l \in \mathcal{L}\backslash \mathcal{L}^*} Pr[l \leadsto P]$), respectively. Thus, this series of deductions establishes the assertion made in \ref{thm:stat1}.
    Second, the proof of Statement \ref{thm:stat2} follows from the fact that the Planar Isotopic Mechanism is $(\epsilon,\delta)$-location private, which proved to achieve the lower bound in Theorem \ref{thm:utility}.
\end{proof}

\section{Experimental Evaluation}
\label{sec:evaluation}
In this section we present experimental evaluation of our framework. All algorithms were implemented using Python, and all experiments were conducted on an Ubuntu 22.04 machine equipped with a 3.5 GHz Intel i7 CPU and 16 GB of memory.
\subsection{Dataset}
We perform experiments using three real-world datasets sourced from public Foursquare check-in data \cite{foursquare_dataset} spanning from April 2012 to January 2014. Our focus is on three cities — New York City (NYC), Istanbul (IST), and Tokyo (TKY)— chosen for their extensive mobility records. We apply a filter to exclude locations visited fewer than 10 times in each city, then proceed to extract mobility records for individual users.
\paragraph{Trajectory data} To refine the trajectory dataset, we meticulously structure the user's mobility records into coherent trajectories. Within these trajectories, we ensure that the time interval between two consecutive points remains constrained, with a maximum gap of 1 hour.
\paragraph{POI data} To identify POIs, we begin by selecting candidate points based on their proximity to each other within a specified distance threshold. We then evaluate these candidates over time, forming clusters of significant stays. A stay is deemed significant if it lasts for a minimum duration and has a minimum number of points within its vicinity. By iteratively refining the candidate selection and clustering process, we effectively capture meaningful spatial locations, which are returned as POIs. For the details of the implementation of the aforementioned approach, please refer to Algorithm \ref{alg:poi_extraction} given Appendix \ref{app:ident_poi}. In our experiments, a location qualifies as a POI if it meets the following conditions: at least two points within the dataset have an elapsed time of 45 minutes or more, and the maximum distance between these points is 250 meters.

\subsection{Location Prediction Model}
To model the adversary's belief regarding the user's location, we leverage the Flashback deep learning method \cite{yang2020location} \footnote{We use the implementation available on \url{https://github.com/eXascaleInfolab/Flashback_code}} to accurately forecast the user's whereabouts at a specific timestamp based on observed previous locations. We employ Flashback, a well-established technique, which we adapt to our needs by enhancing its functionality. Unlike the standard implementation, our modified version of Flashback doesn't solely output the locations with the highest scores. Instead, it predicts the 100 most probable locations for the user at the given time step. Utilizing the same methodology, we build a model for predicting visited Points of Interest (POIs) based on the time intervals between two consecutive location updates. This model capitalizes on the compilation of previously extracted POIs, coupled with the average duration users spend at these locations.

\subsection{Conducted Evaluations}
We aim to assess the influence of $\delta$ and $\epsilon$ values on location-based privacy loss, trajectory-based privacy loss, and Point of Interest (POI) privacy loss. 

\begin{figure}[h]
    \centering
    \begin{subfigure}{0.23\textwidth}
        \includegraphics[width=\linewidth]{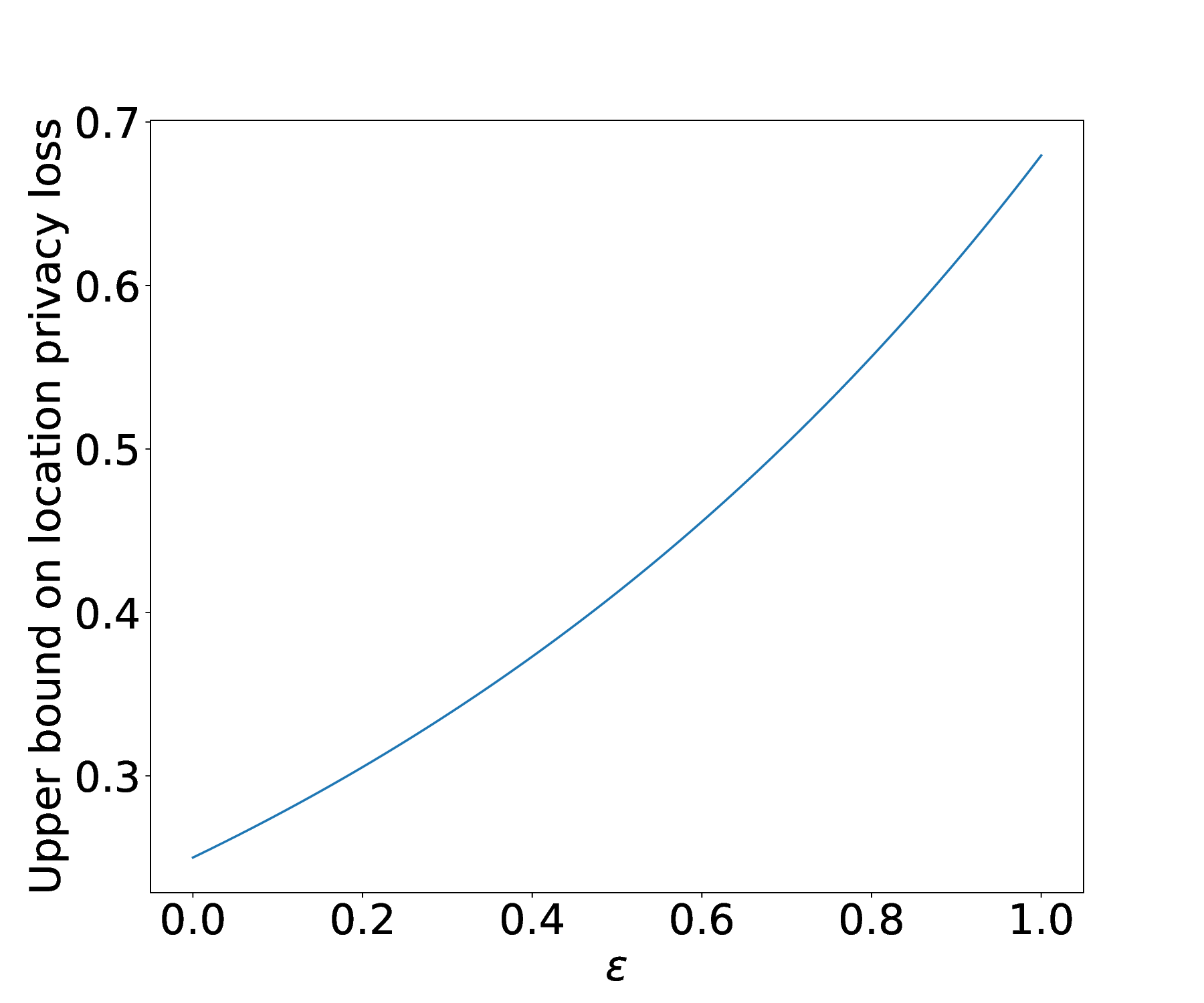}
        \caption{Upper bound as a function of $\epsilon$ ($\delta = 0.8$)}
        \label{fig:epsilon_location_}
    \end{subfigure}
    \hfill
    \begin{subfigure}{0.23\textwidth}
        \includegraphics[width=\linewidth]{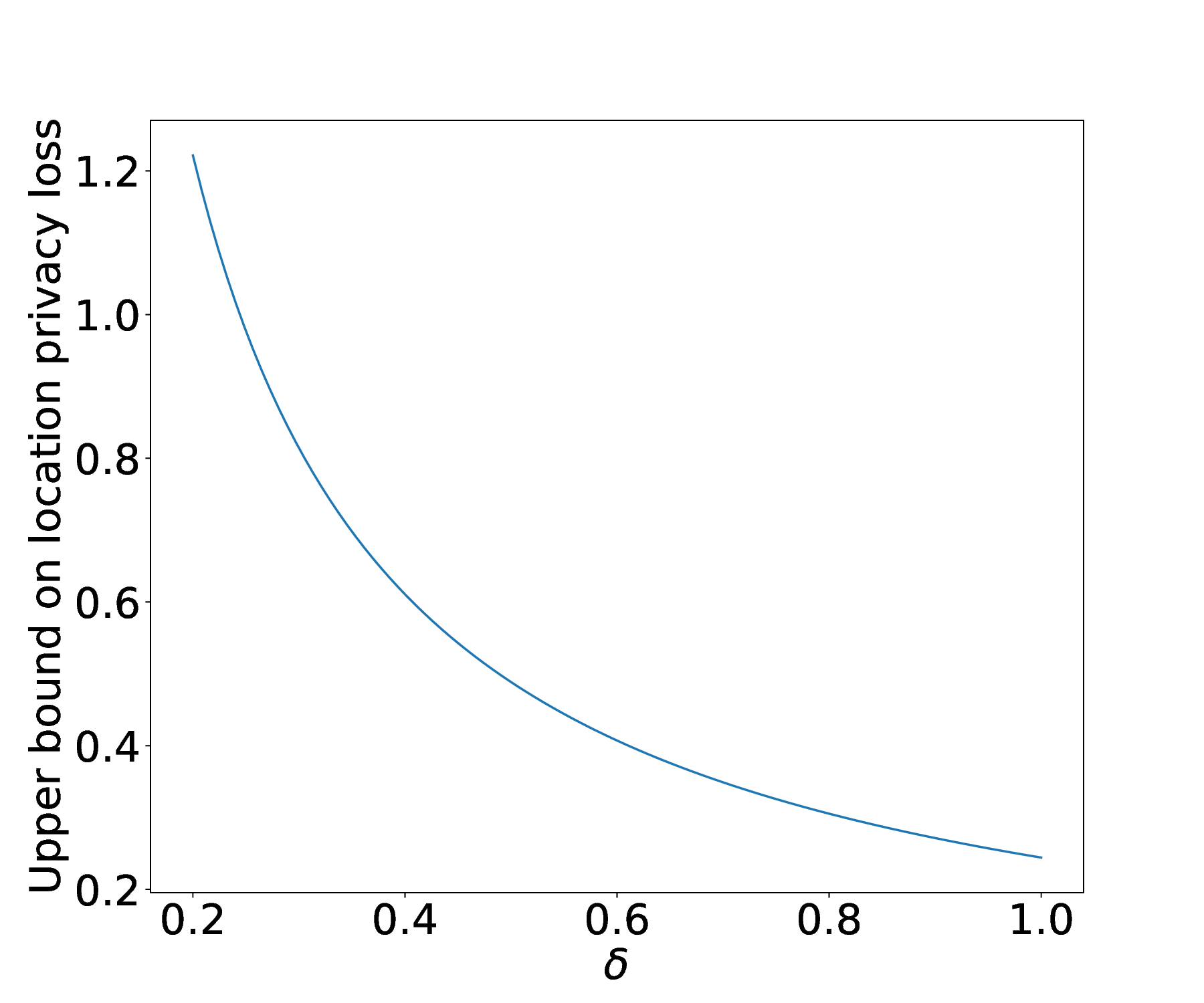}
        \caption{Upper bound as a function of $\delta$ ($\epsilon = 0.2$)}
        \label{fig:delta_location_}
    \end{subfigure}
    \caption{Upper bound on location privacy loss}
    \label{fig:epsilon_delta_location}
\end{figure}

Figure \ref{fig:epsilon_delta_location} illustrates the upper bound on location privacy loss as a function of $\epsilon$ and $\delta$. In Figure \ref{fig:epsilon_location_}, which shows the upper bound as a function of $\epsilon$ with $\delta$ fixed at 0.8, we observe a clear sub-exponential increase in the upper bound as $\epsilon$ grows. Conversely, Figure \ref{fig:delta_location_} depicts the upper bound as a function of $\delta$ with $\epsilon$ held constant at 0.2. Here, we observe a gradual decrease in the upper bound as $\delta$ increases, suggesting that larger values of $\delta$ lead to looser bounds on location privacy loss. 

In addition, we investigate the impact of the parameters $\epsilon$ and $\delta$ on the privacy of trajectory data. We conduct our analysis by randomly selecting a trajectory comprising 12 released locations. Subsequently, each released location $z_i$ undergoes perturbation using the Planar Isotropic Mechanism with a fixed $\delta$ value of 0.8. The values of $\epsilon_i$ and the set of $\delta$ obfuscation, denoted as $\mathcal{L}^*_i$, employed at each time step are determined such that the introduced error remains within a threshold of 1000 meters, as defined by Formula \ref{eq:error}.

\begin{figure}[h]
    \centering
    \begin{subfigure}{0.22\textwidth}
        \includegraphics[width=\linewidth]{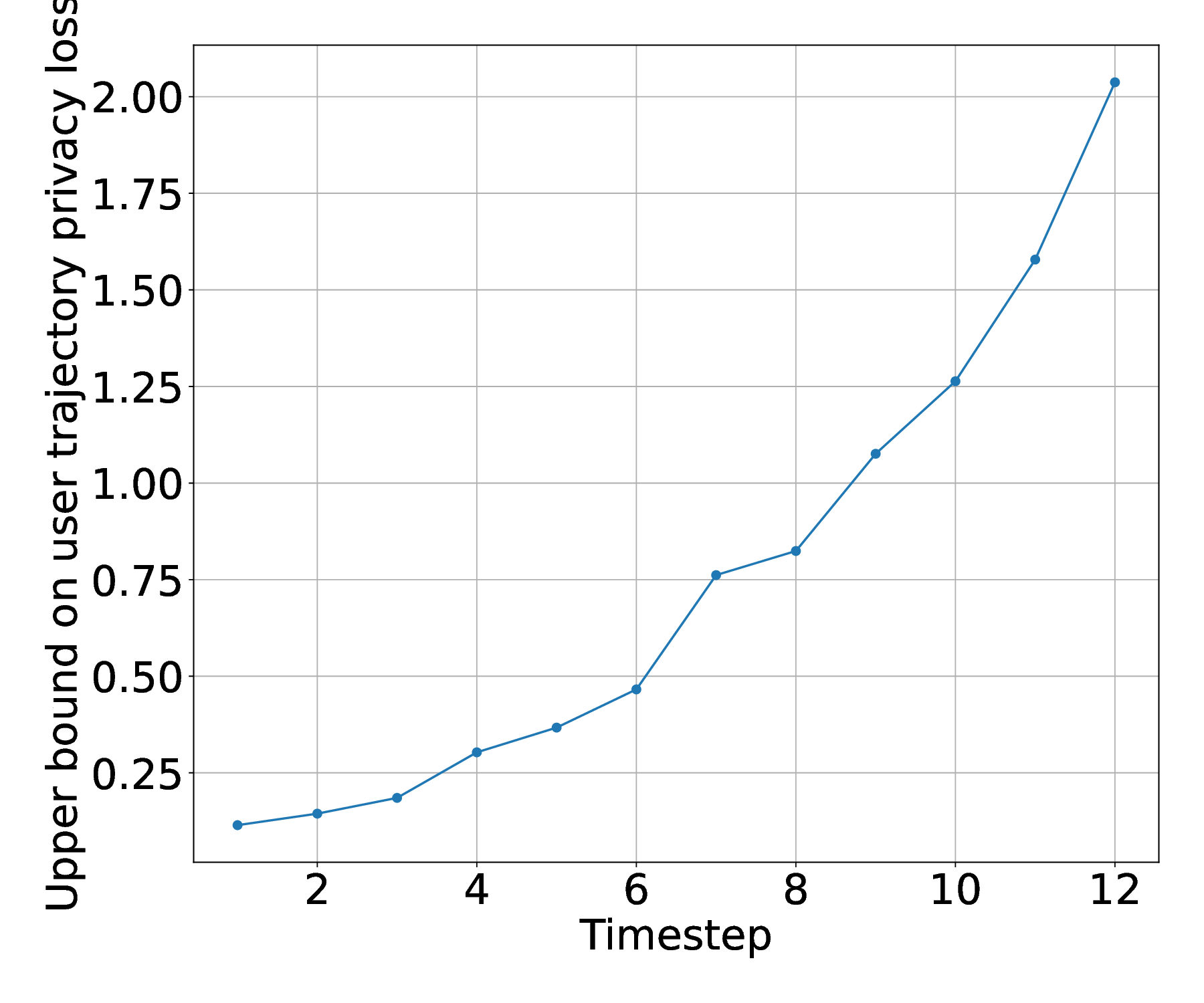}
        \caption{Upper bound on user's trajectory privacy loss}
        \label{fig:delta_location}
    \end{subfigure}
    \hfill
    \begin{subfigure}{0.22\textwidth}
        \includegraphics[width=\linewidth]{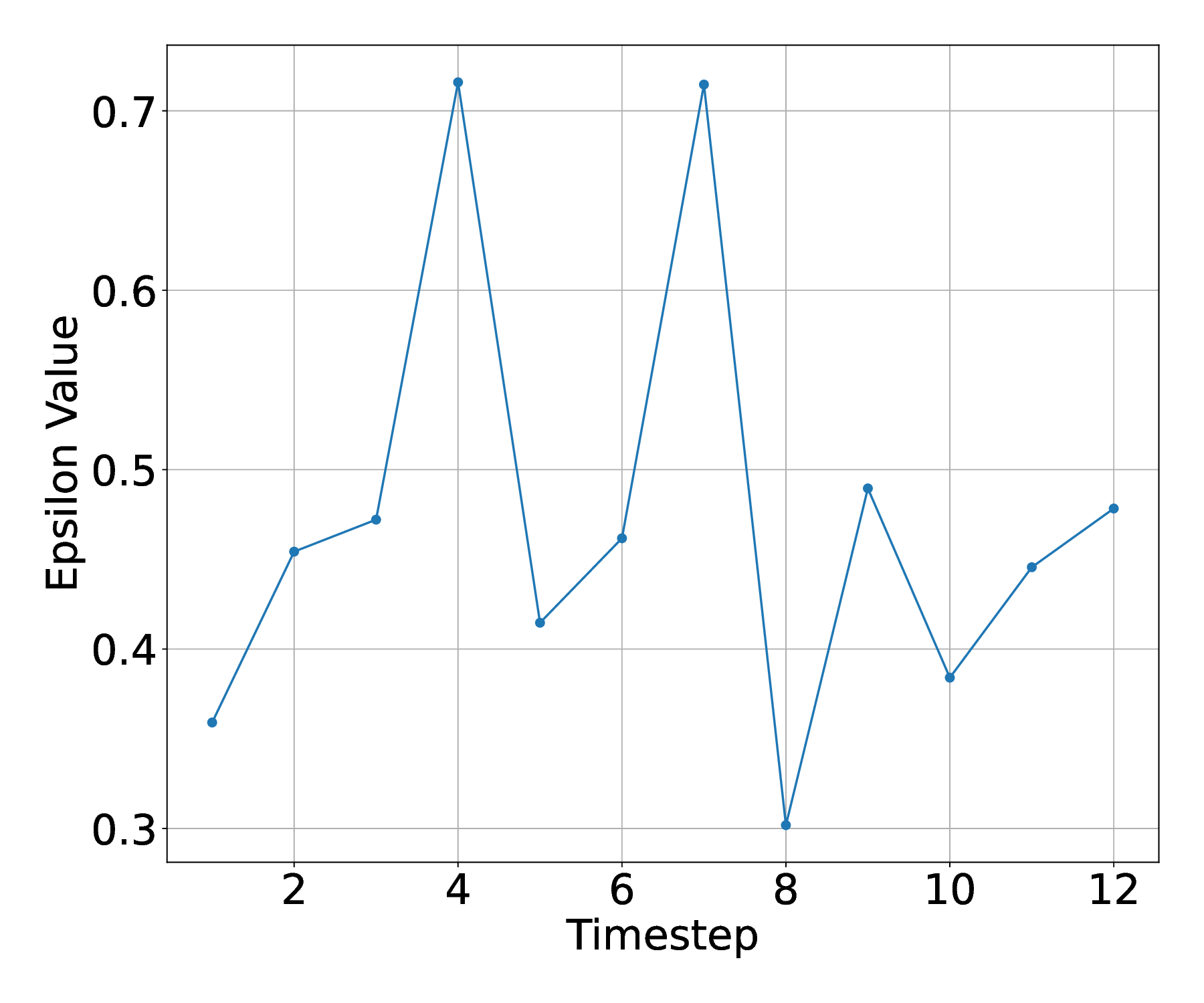}
        \caption{The value of computed $\epsilon_i$ ($\mathcal{E} \leq 1000$)}
        \label{fig:epsilon_location}
    \end{subfigure}
    \hfill
    \begin{subfigure}{0.24\textwidth}
        \includegraphics[width=\linewidth]{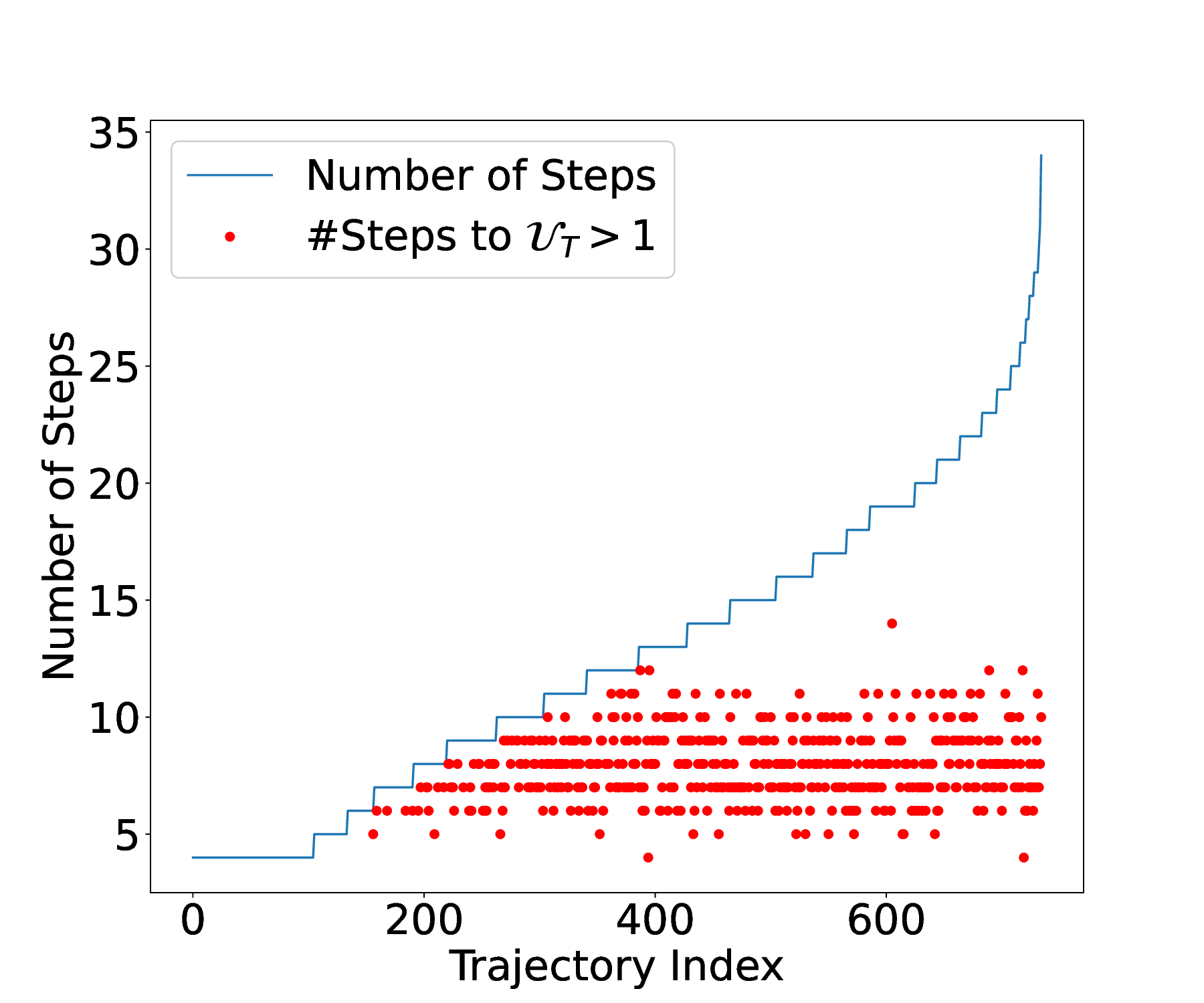}
        \caption{Number of steps required for $\mathcal{U}_T > 1$}
        \label{fig:upper_trajectory}
    \end{subfigure}
    \caption{Trajectory privacy loss quantification}
    \label{fig:trajectory}
\end{figure}

The results reported in Figure \ref{fig:trajectory} illustrates the privacy implications on user's trajectory identification when releasing location data over a series of time steps. In Figure \ref{fig:delta_location}, we observe the upper bound on the user trajectory privacy loss increasing steadily with each timestep. Notably, around timestep 9, the privacy loss surpasses a value of 1. This indicates that the adversary's ability to infer the user's trajectory significantly escalates, potentially enabling full identification of the user's trajectory. Figure \ref{fig:epsilon_location}, shows the variability in the chosen $\epsilon$ values across different time steps in order to fulfill the utility requirement $\mathcal{E} < 1$ km. Finally, Figure \ref{fig:upper_trajectory} reports for each trajectory, the number of perturbed released locations required so that the upper bound probability of identifying the trajectory, denoted as $\mathcal{U}_T$ exceeds 1. In this experiment, we use $\delta = 0.8$ and for each timestep, we compute $\epsilon$ such that $\mathcal{E} < 1$ km.

Furthermore, we leverage our proposed metric for quantifying POI-based privacy loss to evaluate the degree of privacy compromise alongside the associated accuracy under varying values of $\delta$ and $\epsilon$.

\begin{figure}[h]
    \centering
    \begin{subfigure}{0.23\textwidth}
        \includegraphics[width=\linewidth]{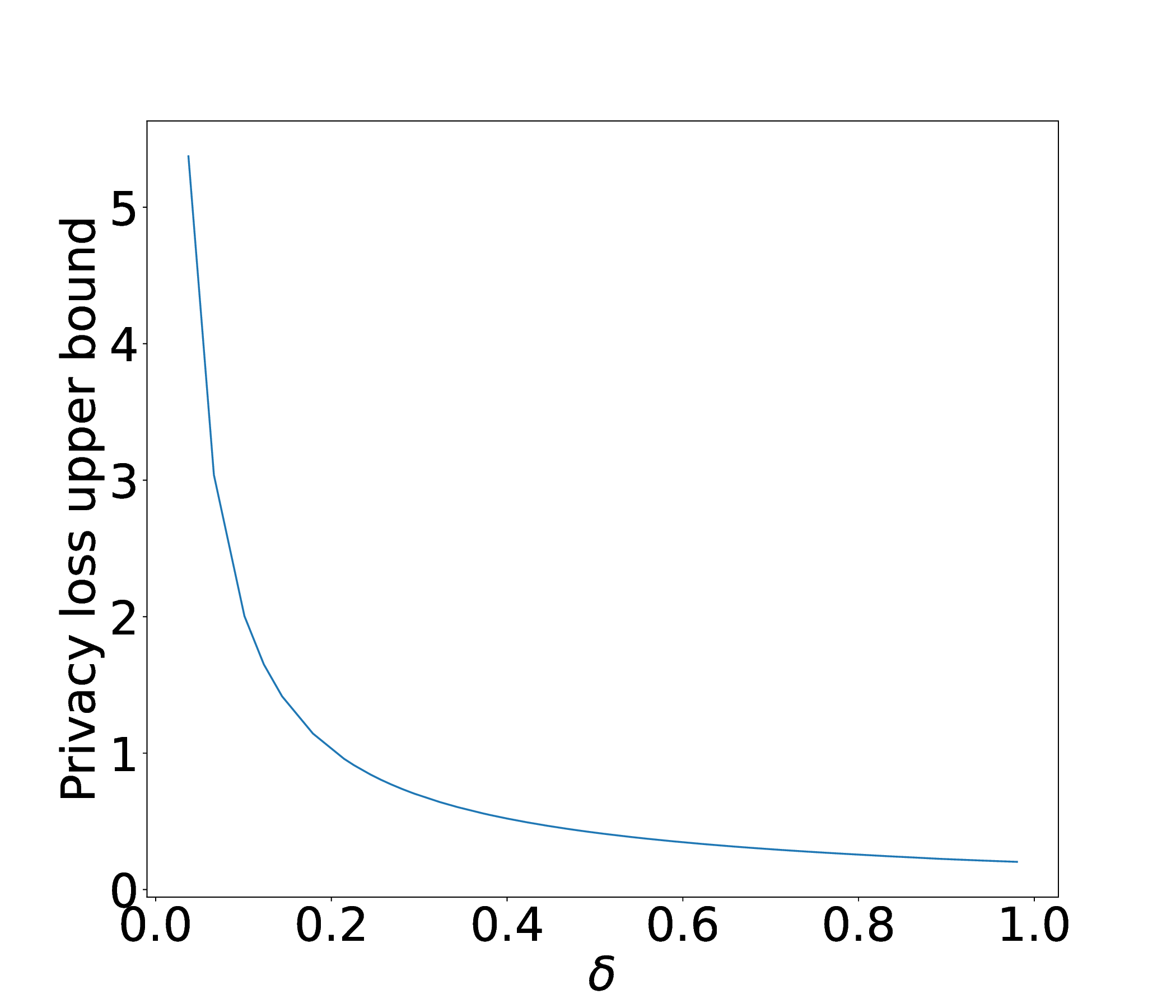}
        \caption{Upper Bound on privacy loss on a visited POI}
        \label{fig:delta_POI}
    \end{subfigure}
    \hfill
    \begin{subfigure}{0.23\textwidth}
        \includegraphics[width=\linewidth]{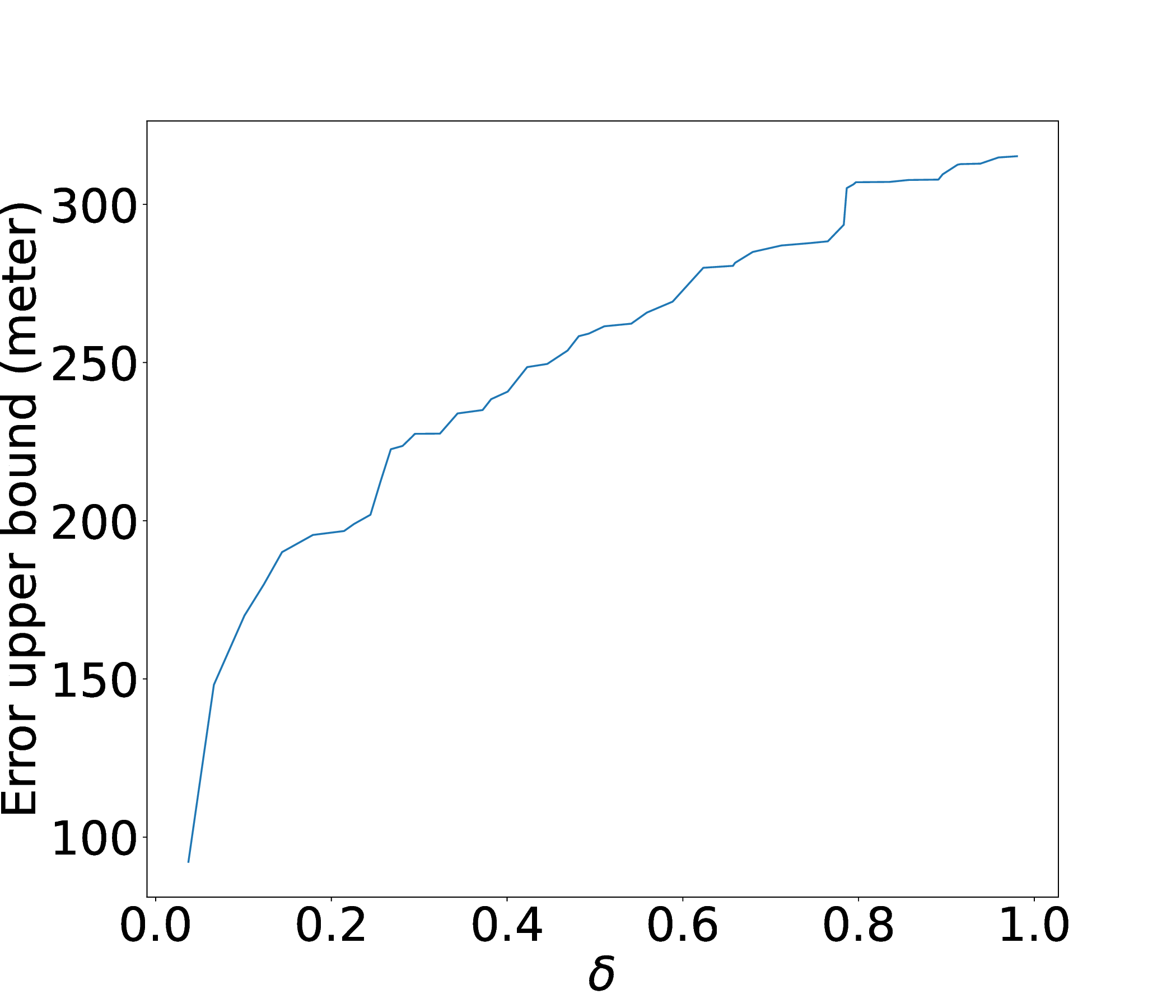}
        \caption{Upper bound on the introduced Error}
        \label{fig:error_delta_POI}
    \end{subfigure}
    \hfill
    \begin{subfigure}{0.23\textwidth}
        \includegraphics[width=\linewidth]{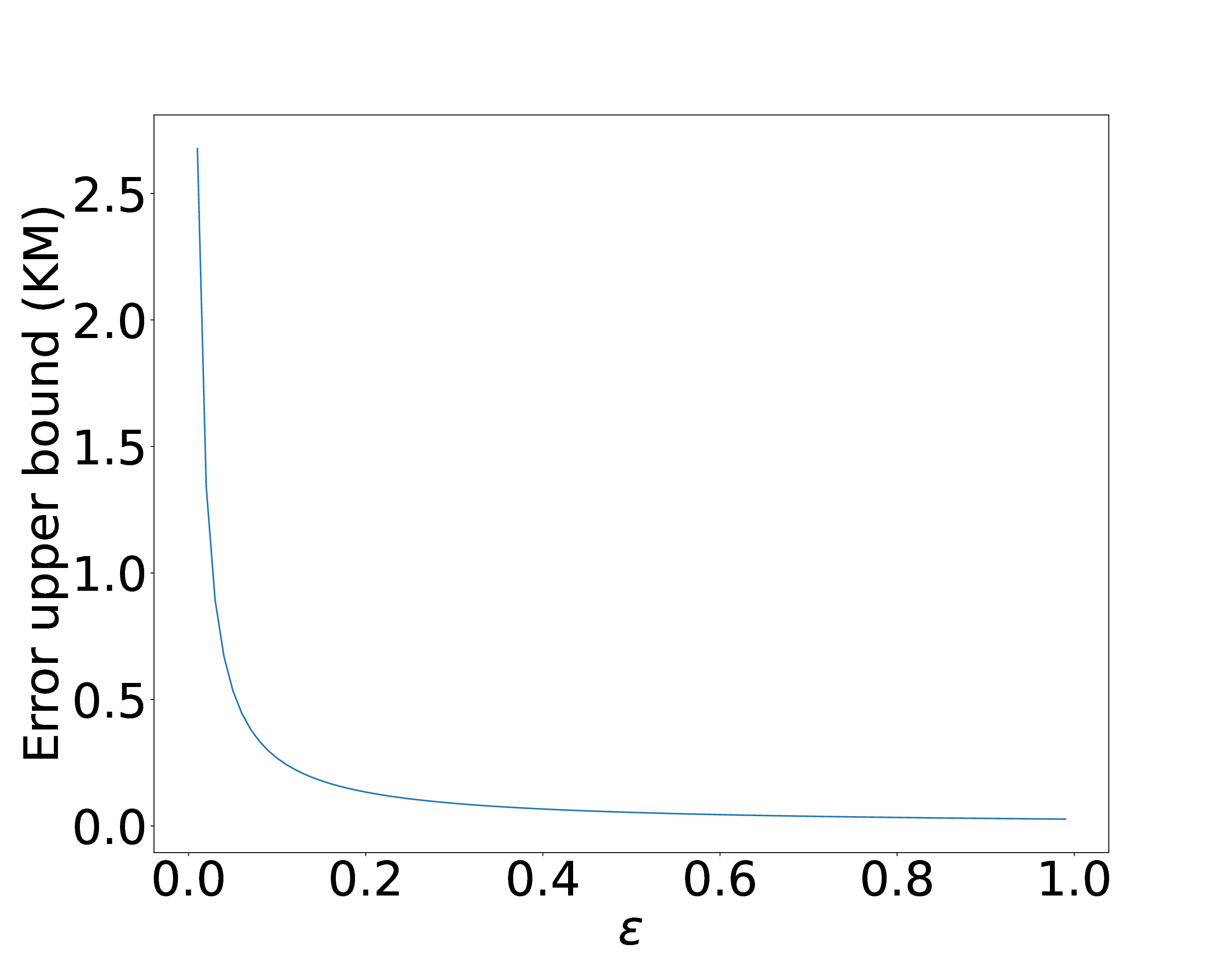}
        \caption{Upper bound on the introduced Error}
        \label{fig:epsilon_POI}
    \end{subfigure}
    \begin{subfigure}{0.23\textwidth}
        \includegraphics[width=\linewidth]{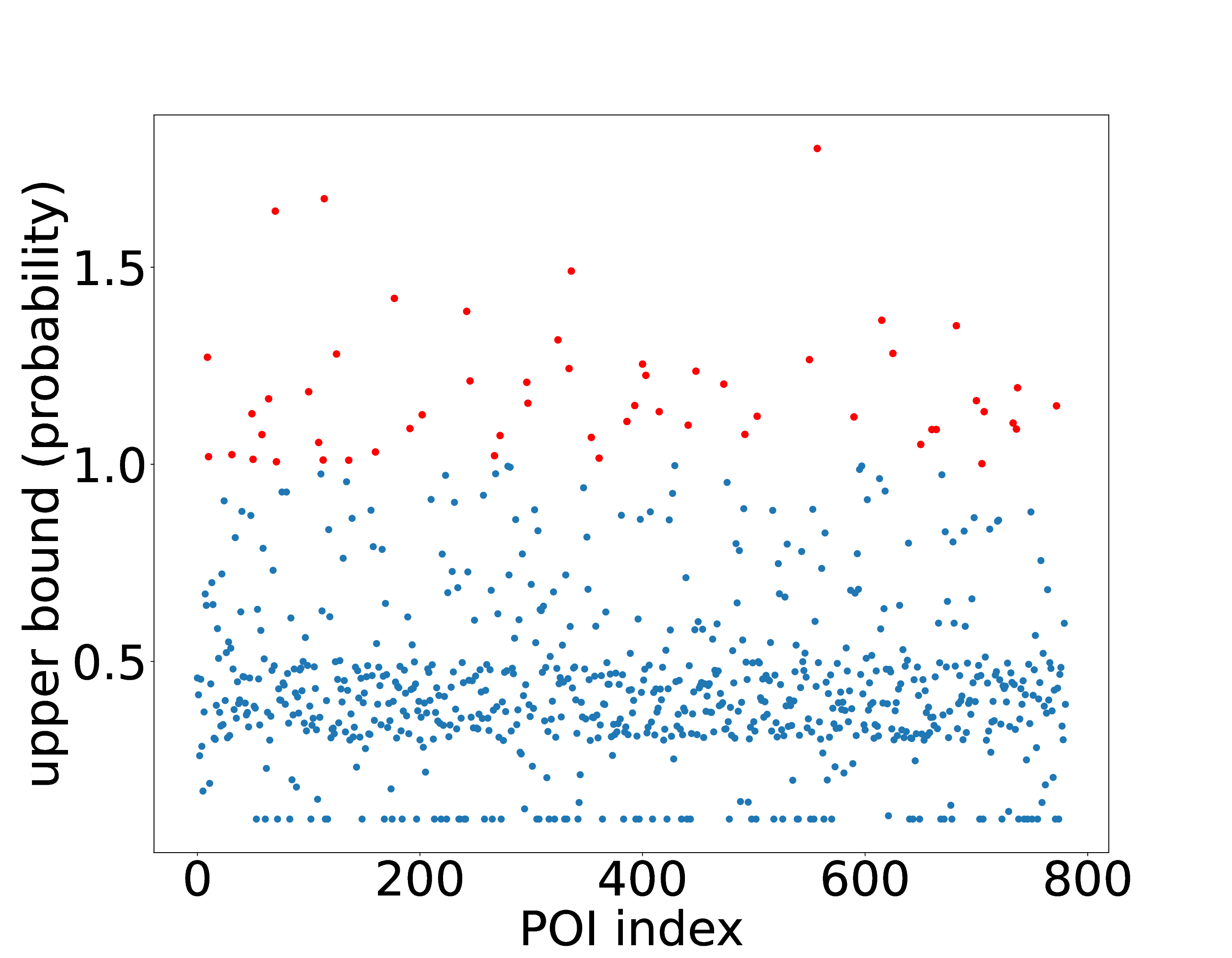}
        \caption{Upper bound on privacy loss for all visited POIs}
        \label{fig:pois_privacy_loss}
    \end{subfigure}
    \caption{POI privacy loss quantification}
    \label{fig:POI}
\end{figure}

We first investigate the impact of varying the parameter $\delta$ on both the upper bound of privacy loss and the upper bound of accuracy for a randomly selected Point of Interest (POI), as elucidated by Theorems \ref{thm:link-LP-POIP} and \ref{thm:utility} respectively. Figure \ref{fig:delta_POI} depicts the relationship between the upper bound of privacy loss on the selected POI and the values of $\delta$ (the value of $\epsilon$ is fixed at 0.1). The observed trend reveals that as the value of $\delta$ decreases, the corresponding upper bound of privacy loss on the selected POI increases exponentially. The findings underscore the paramount importance of meticulously selecting the $\delta$-location set utilized to obfuscate the actual position of the user, particularly when evaluating the resulting privacy loss on the visited POIs. In Figure \ref{fig:error_delta_POI}, we illustrate the relationship between $\delta$ and the accuracy, represented by the upper bound error introduced when interacting with Location-Based Services (LBS). A noticeable trend emerges as the value of $\delta$ decreases: the error bound decreases. This phenomenon primarily stems from the diminishing area of the convex hull (CH) encapsulating $\Delta d[\mathcal{L}^* \cup {l^*}]$, as fewer locations are considered. A decrease in the area of the convex hull indicates a more concentrated clustering of potential locations, resulting in the perturbed location released being closer to the user's true location. 

Then, we delve into exploring the impact of varying $\epsilon$ on the upper bound of introduced error while keeping the value of $\delta$ fixed at 0.7. Figure \ref{fig:epsilon_POI} illustrates that as the value of $\epsilon$ increases, the introduced error decreases exponentially. This observation underscores the critical significance of $\epsilon$ in influencing the accuracy of LBS queries. 

We conclude our analysis by quantifying the upper bound of identifying all visited Points of Interest (POIs) within the dataset when disclosing location information following each visited POI, utilizing the Planar Isotropic Mechanism with $\epsilon$ and $\delta$ chosen to ensure an error upper bound of 1 km. Figure 5d presents the upper bound of identifying each visited POI. Out of the 781 considered POIs, the upper bound identification probability exceeds 1 for 44 POIs. This indicates that, according to the considered inference model, the adversary may be capable of fully identifying these 44 visited POIs when enforcing an error upper bound of 1 km.

\section{Conclusion}
\label{sec:conclusion}
In conclusion, this paper presents a comprehensive framework for addressing privacy concerns in LBS by introducing several novel contributions. Through the extension of the $\delta$-location set concept, we have developed refined notions of $(\epsilon, \delta)$-location privacy, $(\epsilon, \delta, \theta)$-trajectory privacy, and $(\epsilon, \delta, \theta)$-POI privacy, offering a nuanced understanding of privacy leakages across different dimensions of location-based interactions. By establishing fundamental connections between these privacy notions, we provide a holistic approach to privacy preservation in LBS, enabling reasoned analysis and mitigation strategies. Furthermore, our lower bound analysis offers insights into the utility-privacy trade-offs, facilitating the design and evaluation of privacy-preserving mechanisms. Finally, through the instantiation of our framework with the Plannar Isotopic Mechanism, we demonstrate the practical applicability of our approach while ensuring optimal utility and quantifying privacy leakages.
We believe that our contributions substantially advance the comprehension and practical application of privacy-preserving mechanisms in Location-Based Services (LBS). Our work presents a comprehensive framework that effectively balances utility with the complex array of privacy concerns, providing invaluable insights to support informed decision-making across various location-based real-world scenarios.

Nevertheless, it is imperative to acknowledge that the optimization of privacy within the framework of sequential trajectories, encompassing visited POIs, presents an unresolved challenge. Specifically, the determination of suitable parameter values, such as $\epsilon$ and $\delta$, during each interaction with the LBS to attain optimal privacy preservation necessitates additional scrutiny. This issue constitutes a focal point for forthcoming researches.

\begin{acks}
This work were supported by the french national research agency funded
project AUTOPSY (grant no. ANR-20-CYAL-0008).
\end{acks}

\bibliographystyle{ACM-Reference-Format}
\bibliography{sample-base}

%%% -*-BibTeX-*-
%%% Do NOT edit. File created by BibTeX with style
%%% ACM-Reference-Format-Journals [18-Jan-2012].

\begin{thebibliography}{26}

%%% ====================================================================
%%% NOTE TO THE USER: you can override these defaults by providing
%%% customized versions of any of these macros before the \bibliography
%%% command.  Each of them MUST provide its own final punctuation,
%%% except for \shownote{}, \showDOI{}, and \showURL{}.  The latter two
%%% do not use final punctuation, in order to avoid confusing it with
%%% the Web address.
%%%
%%% To suppress output of a particular field, define its macro to expand
%%% to an empty string, or better, \unskip, like this:
%%%
%%% \newcommand{\showDOI}[1]{\unskip}   % LaTeX syntax
%%%
%%% \def \showDOI #1{\unskip}           % plain TeX syntax
%%%
%%% ====================================================================

\ifx \showCODEN    \undefined \def \showCODEN     #1{\unskip}     \fi
\ifx \showDOI      \undefined \def \showDOI       #1{#1}\fi
\ifx \showISBNx    \undefined \def \showISBNx     #1{\unskip}     \fi
\ifx \showISBNxiii \undefined \def \showISBNxiii  #1{\unskip}     \fi
\ifx \showISSN     \undefined \def \showISSN      #1{\unskip}     \fi
\ifx \showLCCN     \undefined \def \showLCCN      #1{\unskip}     \fi
\ifx \shownote     \undefined \def \shownote      #1{#1}          \fi
\ifx \showarticletitle \undefined \def \showarticletitle #1{#1}   \fi
\ifx \showURL      \undefined \def \showURL       {\relax}        \fi
% The following commands are used for tagged output and should be
% invisible to TeX
\providecommand\bibfield[2]{#2}
\providecommand\bibinfo[2]{#2}
\providecommand\natexlab[1]{#1}
\providecommand\showeprint[2][]{arXiv:#2}

\bibitem[Andr{\'e}s et~al\mbox{.}(2013)]%
        {andres2013geo}
\bibfield{author}{\bibinfo{person}{Miguel~E Andr{\'e}s},
  \bibinfo{person}{Nicol{\'a}s~E Bordenabe}, \bibinfo{person}{Konstantinos
  Chatzikokolakis}, {and} \bibinfo{person}{Catuscia Palamidessi}.}
  \bibinfo{year}{2013}\natexlab{}.
\newblock \showarticletitle{Geo-indistinguishability: Differential privacy for
  location-based systems}. In \bibinfo{booktitle}{\emph{Proceedings of the 2013
  ACM SIGSAC conference on Computer \& communications security}}.
  \bibinfo{pages}{901--914}.
\newblock


\bibitem[Dwork et~al\mbox{.}(2016)]%
        {dwork2016calibrating}
\bibfield{author}{\bibinfo{person}{Cynthia Dwork}, \bibinfo{person}{Frank
  McSherry}, \bibinfo{person}{Kobbi Nissim}, {and} \bibinfo{person}{Adam
  Smith}.} \bibinfo{year}{2016}\natexlab{}.
\newblock \showarticletitle{Calibrating noise to sensitivity in private data
  analysis}.
\newblock \bibinfo{journal}{\emph{Journal of Privacy and Confidentiality}}
  \bibinfo{volume}{7}, \bibinfo{number}{3} (\bibinfo{year}{2016}),
  \bibinfo{pages}{17--51}.
\newblock


\bibitem[Dwork et~al\mbox{.}(2010)]%
        {dwork2010differential}
\bibfield{author}{\bibinfo{person}{Cynthia Dwork}, \bibinfo{person}{Moni Naor},
  \bibinfo{person}{Toniann Pitassi}, {and} \bibinfo{person}{Guy~N Rothblum}.}
  \bibinfo{year}{2010}\natexlab{}.
\newblock \showarticletitle{Differential privacy under continual observation}.
  In \bibinfo{booktitle}{\emph{Proceedings of the forty-second ACM symposium on
  Theory of computing}}. \bibinfo{pages}{715--724}.
\newblock


\bibitem[Eckhoff et~al\mbox{.}(2011)]%
        {eckhoff2011slotswap}
\bibfield{author}{\bibinfo{person}{David Eckhoff}, \bibinfo{person}{Reinhard
  German}, \bibinfo{person}{Christoph Sommer}, \bibinfo{person}{Falko
  Dressler}, {and} \bibinfo{person}{Tobias Gansen}.}
  \bibinfo{year}{2011}\natexlab{}.
\newblock \showarticletitle{Slotswap: Strong and affordable location privacy in
  intelligent transportation systems}.
\newblock \bibinfo{journal}{\emph{IEEE Communications Magazine}}
  \bibinfo{volume}{49}, \bibinfo{number}{11} (\bibinfo{year}{2011}),
  \bibinfo{pages}{126--133}.
\newblock


\bibitem[Fang and Chang(2014)]%
        {fang2014differential}
\bibfield{author}{\bibinfo{person}{Chengfang Fang} {and}
  \bibinfo{person}{Ee-Chien Chang}.} \bibinfo{year}{2014}\natexlab{}.
\newblock \showarticletitle{Differential privacy with $\delta$-neighbourhood
  for spatial and dynamic datasets}. In \bibinfo{booktitle}{\emph{Proceedings
  of the 9th ACM symposium on Information, computer and communications
  security}}. \bibinfo{pages}{159--170}.
\newblock


\bibitem[Foursquare({[n.\,d.]})]%
        {foursquare_dataset}
\bibfield{author}{\bibinfo{person}{Foursquare}.}
  \bibinfo{year}{[n.\,d.]}\natexlab{}.
\newblock \bibinfo{title}{Foursquare Dataset}.
\newblock
  \bibinfo{howpublished}{\url{https://sites.google.com/site/yangdingqi/home/foursquare-dataset}}.
\newblock
\newblock
\shownote{Accessed: April 16th, 2024}.


\bibitem[Gruteser and Hoh(2005)]%
        {gruteser2005anonymity}
\bibfield{author}{\bibinfo{person}{Marco Gruteser} {and} \bibinfo{person}{Baik
  Hoh}.} \bibinfo{year}{2005}\natexlab{}.
\newblock \showarticletitle{On the anonymity of periodic location samples}. In
  \bibinfo{booktitle}{\emph{International Conference on Security in Pervasive
  Computing}}. Springer, \bibinfo{pages}{179--192}.
\newblock


\bibitem[He et~al\mbox{.}(2017)]%
        {he2017category}
\bibfield{author}{\bibinfo{person}{Jing He}, \bibinfo{person}{Xin Li}, {and}
  \bibinfo{person}{Lejian Liao}.} \bibinfo{year}{2017}\natexlab{}.
\newblock \showarticletitle{Category-aware next point-of-interest
  recommendation via listwise bayesian personalized ranking.}. In
  \bibinfo{booktitle}{\emph{IJCAI}}, Vol.~\bibinfo{volume}{17}.
  \bibinfo{pages}{1837--1843}.
\newblock


\bibitem[He et~al\mbox{.}(2016)]%
        {he2016inferring}
\bibfield{author}{\bibinfo{person}{Jing He}, \bibinfo{person}{Xin Li},
  \bibinfo{person}{Lejian Liao}, \bibinfo{person}{Dandan Song}, {and}
  \bibinfo{person}{William Cheung}.} \bibinfo{year}{2016}\natexlab{}.
\newblock \showarticletitle{Inferring a personalized next point-of-interest
  recommendation model with latent behavior patterns}. In
  \bibinfo{booktitle}{\emph{Proceedings of the AAAI Conference on Artificial
  Intelligence}}, Vol.~\bibinfo{volume}{30}.
\newblock


\bibitem[Hoh et~al\mbox{.}(2010)]%
        {hoh2010achieving}
\bibfield{author}{\bibinfo{person}{Baik Hoh}, \bibinfo{person}{Marco Gruteser},
  \bibinfo{person}{Hui Xiong}, {and} \bibinfo{person}{Ansaf Alrabady}.}
  \bibinfo{year}{2010}\natexlab{}.
\newblock \showarticletitle{Achieving guaranteed anonymity in gps traces via
  uncertainty-aware path cloaking}.
\newblock \bibinfo{journal}{\emph{IEEE Transactions on Mobile Computing}}
  \bibinfo{volume}{9}, \bibinfo{number}{8} (\bibinfo{year}{2010}),
  \bibinfo{pages}{1089--1107}.
\newblock


\bibitem[Jiang et~al\mbox{.}(2021)]%
        {jiang2021location}
\bibfield{author}{\bibinfo{person}{Hongbo Jiang}, \bibinfo{person}{Jie Li},
  \bibinfo{person}{Ping Zhao}, \bibinfo{person}{Fanzi Zeng},
  \bibinfo{person}{Zhu Xiao}, {and} \bibinfo{person}{Arun Iyengar}.}
  \bibinfo{year}{2021}\natexlab{}.
\newblock \showarticletitle{Location privacy-preserving mechanisms in
  location-based services: A comprehensive survey}.
\newblock \bibinfo{journal}{\emph{ACM Computing Surveys (CSUR)}}
  \bibinfo{volume}{54}, \bibinfo{number}{1} (\bibinfo{year}{2021}),
  \bibinfo{pages}{1--36}.
\newblock


\bibitem[Jiang et~al\mbox{.}(2018)]%
        {jiang2018roblop}
\bibfield{author}{\bibinfo{person}{Hongbo Jiang}, \bibinfo{person}{Ping Zhao},
  {and} \bibinfo{person}{Chen Wang}.} \bibinfo{year}{2018}\natexlab{}.
\newblock \showarticletitle{RobLoP: Towards robust privacy preserving against
  location dependent attacks in continuous LBS queries}.
\newblock \bibinfo{journal}{\emph{IEEE/ACM Transactions on Networking}}
  \bibinfo{volume}{26}, \bibinfo{number}{2} (\bibinfo{year}{2018}),
  \bibinfo{pages}{1018--1032}.
\newblock


\bibitem[Kasiviswanathan et~al\mbox{.}(2011)]%
        {kasiviswanathan2011can}
\bibfield{author}{\bibinfo{person}{Shiva~Prasad Kasiviswanathan},
  \bibinfo{person}{Homin~K Lee}, \bibinfo{person}{Kobbi Nissim},
  \bibinfo{person}{Sofya Raskhodnikova}, {and} \bibinfo{person}{Adam Smith}.}
  \bibinfo{year}{2011}\natexlab{}.
\newblock \showarticletitle{What can we learn privately?}
\newblock \bibinfo{journal}{\emph{SIAM J. Comput.}} \bibinfo{volume}{40},
  \bibinfo{number}{3} (\bibinfo{year}{2011}), \bibinfo{pages}{793--826}.
\newblock


\bibitem[Kellaris et~al\mbox{.}(2014)]%
        {kellaris2014differentially}
\bibfield{author}{\bibinfo{person}{Georgios Kellaris}, \bibinfo{person}{Stavros
  Papadopoulos}, \bibinfo{person}{Xiaokui Xiao}, {and}
  \bibinfo{person}{Dimitris Papadias}.} \bibinfo{year}{2014}\natexlab{}.
\newblock \showarticletitle{Differentially private event sequences over
  infinite streams}.
\newblock \bibinfo{journal}{\emph{Proceedings of the VLDB Endowment}}
  \bibinfo{volume}{7}, \bibinfo{number}{12} (\bibinfo{year}{2014}),
  \bibinfo{pages}{1155--1166}.
\newblock


\bibitem[Meehan and Chaudhuri(2021)]%
        {meehan2021location}
\bibfield{author}{\bibinfo{person}{Casey Meehan} {and}
  \bibinfo{person}{Kamalika Chaudhuri}.} \bibinfo{year}{2021}\natexlab{}.
\newblock \showarticletitle{Location trace privacy under conditional priors}.
  In \bibinfo{booktitle}{\emph{International Conference on Artificial
  Intelligence and Statistics}}. PMLR, \bibinfo{pages}{2881--2889}.
\newblock


\bibitem[Palanisamy and Liu(2014)]%
        {palanisamy2014attack}
\bibfield{author}{\bibinfo{person}{Balaji Palanisamy} {and}
  \bibinfo{person}{Ling Liu}.} \bibinfo{year}{2014}\natexlab{}.
\newblock \showarticletitle{Attack-resilient mix-zones over road networks:
  architecture and algorithms}.
\newblock \bibinfo{journal}{\emph{IEEE Transactions on Mobile Computing}}
  \bibinfo{volume}{14}, \bibinfo{number}{3} (\bibinfo{year}{2014}),
  \bibinfo{pages}{495--508}.
\newblock


\bibitem[Palanisamy et~al\mbox{.}(2014)]%
        {palanisamy2014anonymizing}
\bibfield{author}{\bibinfo{person}{Balaji Palanisamy}, \bibinfo{person}{Ling
  Liu}, \bibinfo{person}{Kisung Lee}, \bibinfo{person}{Shicong Meng},
  \bibinfo{person}{Yuzhe Tang}, {and} \bibinfo{person}{Yang Zhou}.}
  \bibinfo{year}{2014}\natexlab{}.
\newblock \showarticletitle{Anonymizing continuous queries with delay-tolerant
  mix-zones over road networks}.
\newblock \bibinfo{journal}{\emph{Distributed and Parallel Databases}}
  \bibinfo{volume}{32} (\bibinfo{year}{2014}), \bibinfo{pages}{91--118}.
\newblock


\bibitem[Primault et~al\mbox{.}(2018)]%
        {primault2018long}
\bibfield{author}{\bibinfo{person}{Vincent Primault}, \bibinfo{person}{Antoine
  Boutet}, \bibinfo{person}{Sonia~Ben Mokhtar}, {and} \bibinfo{person}{Lionel
  Brunie}.} \bibinfo{year}{2018}\natexlab{}.
\newblock \showarticletitle{The long road to computational location privacy: A
  survey}.
\newblock \bibinfo{journal}{\emph{IEEE Communications Surveys \& Tutorials}}
  \bibinfo{volume}{21}, \bibinfo{number}{3} (\bibinfo{year}{2018}),
  \bibinfo{pages}{2772--2793}.
\newblock


\bibitem[Rastogi et~al\mbox{.}(2009)]%
        {rastogi2009relationship}
\bibfield{author}{\bibinfo{person}{Vibhor Rastogi}, \bibinfo{person}{Michael
  Hay}, \bibinfo{person}{Gerome Miklau}, {and} \bibinfo{person}{Dan Suciu}.}
  \bibinfo{year}{2009}\natexlab{}.
\newblock \showarticletitle{Relationship privacy: output perturbation for
  queries with joins}. In \bibinfo{booktitle}{\emph{Proceedings of the
  twenty-eighth ACM SIGMOD-SIGACT-SIGART symposium on Principles of database
  systems}}. \bibinfo{pages}{107--116}.
\newblock


\bibitem[Shin et~al\mbox{.}(2012)]%
        {shin2012privacy}
\bibfield{author}{\bibinfo{person}{Kang~G Shin}, \bibinfo{person}{Xiaoen Ju},
  \bibinfo{person}{Zhigang Chen}, {and} \bibinfo{person}{Xin Hu}.}
  \bibinfo{year}{2012}\natexlab{}.
\newblock \showarticletitle{Privacy protection for users of location-based
  services}.
\newblock \bibinfo{journal}{\emph{IEEE Wireless Communications}}
  \bibinfo{volume}{19}, \bibinfo{number}{1} (\bibinfo{year}{2012}),
  \bibinfo{pages}{30--39}.
\newblock


\bibitem[Shokri et~al\mbox{.}(2011)]%
        {shokri2011quantifying}
\bibfield{author}{\bibinfo{person}{Reza Shokri}, \bibinfo{person}{George
  Theodorakopoulos}, \bibinfo{person}{Jean-Yves Le~Boudec}, {and}
  \bibinfo{person}{Jean-Pierre Hubaux}.} \bibinfo{year}{2011}\natexlab{}.
\newblock \showarticletitle{Quantifying location privacy}. In
  \bibinfo{booktitle}{\emph{2011 IEEE symposium on security and privacy}}.
  IEEE, \bibinfo{pages}{247--262}.
\newblock


\bibitem[Stirbys et~al\mbox{.}(2017)]%
        {stirbys2017privacy}
\bibfield{author}{\bibinfo{person}{Simonas Stirbys}, \bibinfo{person}{Omar~Abu
  Nabah}, \bibinfo{person}{Per Hallgren}, {and} \bibinfo{person}{Andrei
  Sabelfeld}.} \bibinfo{year}{2017}\natexlab{}.
\newblock \showarticletitle{Privacy-preserving location-proximity for mobile
  apps}. In \bibinfo{booktitle}{\emph{2017 25th Euromicro International
  Conference on Parallel, Distributed and Network-based Processing (PDP)}}.
  IEEE, \bibinfo{pages}{337--345}.
\newblock


\bibitem[Sweeney(2002)]%
        {sweeney2002k}
\bibfield{author}{\bibinfo{person}{Latanya Sweeney}.}
  \bibinfo{year}{2002}\natexlab{}.
\newblock \showarticletitle{k-anonymity: A model for protecting privacy}.
\newblock \bibinfo{journal}{\emph{International journal of uncertainty,
  fuzziness and knowledge-based systems}} \bibinfo{volume}{10},
  \bibinfo{number}{05} (\bibinfo{year}{2002}), \bibinfo{pages}{557--570}.
\newblock


\bibitem[Xiao and Xiong(2015)]%
        {xiao2015protecting}
\bibfield{author}{\bibinfo{person}{Yonghui Xiao} {and} \bibinfo{person}{Li
  Xiong}.} \bibinfo{year}{2015}\natexlab{}.
\newblock \showarticletitle{Protecting locations with differential privacy
  under temporal correlations}. In \bibinfo{booktitle}{\emph{Proceedings of the
  22nd ACM SIGSAC Conference on Computer and Communications Security}}.
  \bibinfo{pages}{1298--1309}.
\newblock


\bibitem[Yang et~al\mbox{.}(2017)]%
        {yang2017neural}
\bibfield{author}{\bibinfo{person}{Cheng Yang}, \bibinfo{person}{Maosong Sun},
  \bibinfo{person}{Wayne~Xin Zhao}, \bibinfo{person}{Zhiyuan Liu}, {and}
  \bibinfo{person}{Edward~Y Chang}.} \bibinfo{year}{2017}\natexlab{}.
\newblock \showarticletitle{A neural network approach to jointly modeling
  social networks and mobile trajectories}.
\newblock \bibinfo{journal}{\emph{ACM Transactions on Information Systems
  (TOIS)}} \bibinfo{volume}{35}, \bibinfo{number}{4} (\bibinfo{year}{2017}),
  \bibinfo{pages}{1--28}.
\newblock


\bibitem[Yang et~al\mbox{.}(2020)]%
        {yang2020location}
\bibfield{author}{\bibinfo{person}{Dingqi Yang}, \bibinfo{person}{Benjamin
  Fankhauser}, \bibinfo{person}{Paolo Rosso}, {and} \bibinfo{person}{Philippe
  Cudre-Mauroux}.} \bibinfo{year}{2020}\natexlab{}.
\newblock \showarticletitle{Location prediction over sparse user mobility
  traces using rnns}. In \bibinfo{booktitle}{\emph{Proceedings of the
  twenty-ninth international joint conference on artificial intelligence}}.
  \bibinfo{pages}{2184--2190}.
\newblock


\end{thebibliography}

%%
%% If your work has an appendix, this is the place to put it.
\appendix
\section{Appendix}
\subsection{Proof of Proposition \ref{prop:trajectory_privacy}}
\label{proof:prop:trajectory_privacy}
Proof is by induction. In this proof, we use $Pr_{t}^-$ and $Pr_{t}^+$ to denote the prior and posterior (i.e., after the release of the timestep $t$ perturbed location)  probabilities.  First, let us suppose that the proposition holds for $t = 1$ and show that it holds for $t=2$. Since we suppose that the mechanism $\mathcal{M}_1$ satisfies Location-based $\epsilon_1$-LDP, then according to Definition \ref{def:adv-location_obfuscation}, we have:
\begin{equation}
\label{eq:prop_start_resonning_}
    \forall l_i \in \mathcal{L}_1: Pr^+_1[l_i = l^*_1] \leq e^{\epsilon_1} \cdot Pr^-_1[l_i = l^*_1]
\end{equation}
We multiply both sides of the previous
inequation with $Pr^-[l_i \leadsto T]$ to get
\begin{multline}
    \label{eq:proof_thm_traj_eq_mult}
    \forall l_i \in \mathcal{L}_1, \forall T \in \mathcal{T}: Pr^+_1[l_i = l^*_1] \cdot Pr^-_1[l_i \leadsto T] \leq \\ e^{\epsilon_1} \cdot Pr^-_1[l_i = l^*_1] \cdot Pr^-_1[l_i \leadsto T]
\end{multline}
Since the output of $\mathcal{M}$ does not affect $Pr[l_i \leadsto T]$ and considering that $l_i = l^*_1$ and $l_i \leadsto T$ are independent events, then, based on Equation \ref{eq:proof_thm_traj_eq_mult}, we have:
\begin{multline}
    \forall l_i \in \mathcal{L}_1, \forall T \in \mathcal{T}: Pr^+_1[l_i = l^*_1 \cap l_i \leadsto T] \leq \\ e^{\epsilon_1} \cdot Pr^-_1[l_i = l^*_1 \cap l_i \leadsto T]
\end{multline}
We sum each side of the previous inequaion to get:
\begin{multline}
\label{eq:proof_prop_traj_eq_inter}
     \forall T \in \mathcal{T}: \sum\limits_{l_i \in \mathcal{L}_1} Pr^+_1[l_i = l^*_1 \cap l_i \leadsto T] \leq \\ e^{\epsilon_1} \cdot \sum\limits_{l_i \in \mathcal{L}_1} Pr^-_1[l_i = l^*_1 \cap l_i \leadsto T]
\end{multline}
Since the probability distributions of $l_i = l^*_1$ and $l_i \leadsto T$ are defined over the set $\mathcal{L}_1$, then we have:
\begin{equation}
    \sum_{l_i \in \mathcal{L}_1} Pr^{\{+,-\}}_1[l_i = l^* \cap l_i \leadsto T] = Pr^{\{+,-\}}_1[T = T^*]
\end{equation}
Now, we can use the previous equation to rewrite Inequation \ref{eq:proof_prop_traj_eq_inter} as following:
\begin{multline}
\label{eq:prop_t_1_first}
     \forall T \in \mathcal{T}: Pr^+_1[T = T^*] \leq e^{\epsilon_1} \cdot Pr^-_1[T = T^*]
\end{multline}
Following the same previous reasoning (Inequations \ref{eq:prop_start_resonning_} to \ref{eq:prop_t_1_first}) for t = 2, we can also deduce the following:
\begin{multline}
\label{eq:prop_t_1_Second}
     \forall T \in \mathcal{T}: Pr^+_2[T = T^*] \leq e^{\epsilon_2} \cdot Pr^-_2[T = T^*]
\end{multline}
Since no information is disclosed by the user about his/her location between the timestep $t_1$ and $t_2$, then we have $Pr^+_1[T = T^*] = Pr^-_2[T = T^*]$. The later deduction together with Inequations \ref{eq:prop_t_1_first} and \ref{eq:prop_t_1_Second} allow to deduce the following:
\begin{multline}
\label{eq:prop_t_2_final}
     \forall T \in \mathcal{T}: Pr^+_2[T = T^*] \leq e^{\epsilon_1 + \epsilon_2} \cdot Pr^-_1[T = T^*]
\end{multline}
which proves the proposition for t=2.

Now let us suppose that the proposition holds for t-1. Following the same previous reasoning (Inequations \ref{eq:prop_start_resonning_} to \ref{eq:prop_t_1_first}) we can deduce the following:
\begin{multline}
\label{eq:prop_t-1_first}
     \forall T \in \mathcal{T}: Pr^+_{k-1}[T = T^*] \leq e^{\sum_{k=1}^{t-1}\epsilon_{j}} \cdot Pr^-_1[T = T^*]
\end{multline}
Now, since we suppose that $\mathcal{M}_t$ satisfies Location-based $\epsilon_t$-LDP, then using the same reasoning from Inequations \ref{eq:prop_start_resonning_} to \ref{eq:prop_t_1_first}, we deduce:
\begin{multline}
\label{eq:prop_t-1_first}
     \forall T \in \mathcal{T}: Pr^+_{k}[T = T^*] \leq e^{\epsilon_{k}} \cdot Pr^-_k[T = T^*]
\end{multline}
Finally, since $Pr^-_k[T = T^*] = Pr^+_{k-1}[T = T^*]$, we get:
\begin{multline}
\label{eq:prop_t-1_first}
     \forall T \in \mathcal{T}: Pr^+_{k}[T = T^*] \leq e^{\sum_{k=1}^{t}\epsilon_{k}} \cdot Pr^-_1[T = T^*]
\end{multline}
which concludes the proof.

\subsection{Proof of Theorem \ref{thm:location-privacy}}
\label{proof:theorem_location_privacy}
    To prove the theorem, we need to show that:
    \begin{enumerate}
        \item $\forall l_i \in \mathcal{L}_t \backslash \mathcal{L}^*_t: Pr^+[l_i = l^*_t | l^*_t \in \mathcal{L}^*_t ] = 0$
        \item $\forall l_i \in \mathcal{L}^*_t: Pr^+[l_i = l^*_t| l^*_t \in \mathcal{L}^*_t] \leq \delta^{-1} \cdot e^\epsilon \cdot Pr^-[l_i = l^*_t]$
    \end{enumerate}
    Equation 1 is derived from from the fact that the considered (strong) adversary is supposed to know that the true location $l^*_t$ falls within the set $\mathcal{L}^*_t$. 
    As the set of locations in $\mathcal{L}_t \backslash \mathcal{L}^*_t$ have not been considered by the application of $\mathcal{M}$, we have:
    \begin{equation}
    \nonumber
        \forall l_t \in \mathcal{L}_t \backslash \mathcal{L}^*_t: Pr^+[l_t = l^*_t] = Pr^-[l_t = l^*_t] = 0
    \end{equation}
    To prove the correctness of Equation 2, let us start from the fact that the $\epsilon$-LDP mechanism $\mathcal{M}$ is applied on the set of locations $\mathcal{L}^*_t$ and that the adversary knows that the real location of $l^*_t$ is in $\mathcal{L}^*_t$. Then, according to Definition \ref{def:DP_los}, for all $l_i \in \mathcal{L}^*_t$ we have:
    \begin{equation}
        \label{eq:dp_th1_app}
         Pr^+[l_i = l^*_t | l^*_t \in \mathcal{L}^*_t] \leq e^\epsilon \cdot Pr^-[l_i = l^*_t | l^*_t \in \mathcal{L}^*_t] 
    \end{equation}
    Let us focus on the left side of Equation \ref{eq:dp_th1_app}. We have:
    \begin{equation}
        Pr^-[l_i = l^*_t | l^*_t \in \mathcal{L}^*_t] = \frac{Pr^-[l_i = l^*_t \cap l^*_t \in \mathcal{L}^*_t]}{Pr[l^*_t \in \mathcal{L}^*_t]}
    \end{equation}
    Since we are dealing with $l_i \in \mathcal{L}^*_t$, then we have:
    \begin{equation}
        \label{eq:th1_lasts}
        Pr^-[l_i = l^*_t | l^*_t \in \mathcal{L}^*_t] = \frac{Pr^-[l_i = l^*_t]}{Pr[l^*_t \in \mathcal{L}^*_t]} = \frac{Pr^-[l_i = l^*_t]}{\delta}
    \end{equation}
    Then, based on Equations \ref{eq:dp_th1_app} and \ref{eq:th1_lasts}, we deduce that $\forall l_i \in \mathcal{L}^*_t$ the following holds:
    \begin{equation}
         Pr^+[l_i = l^*_t| l^*_t \in \mathcal{L}^*_t] \leq  \delta^{-1} \cdot e^\epsilon \cdot Pr^-[l_i = l^*_t]
    \end{equation}
    which concludes the proof.

\subsection{Proof of Theorem \ref{thm:link-LP-TP}}
\label{proof:link-LP-TP}
Since we suppose that the mechanism $\mathcal{M}$ satisfies ($\epsilon, \delta$)-location privacy, then according to Theorem \ref{thm:location-privacy}, we have:
\begin{equation}
    \forall l_i \in \mathcal{L}^*_t: Pr^+[l_i = l^*_t | l^*_t \in \mathcal{L}^*_t] \leq \delta^{-1} \cdot e^\epsilon \cdot Pr^-[l_i = l^*_t]
\end{equation}
We multiply both sides of the previous
inequation with $Pr^-[l_i \leadsto T]$ to get
\begin{multline}
    \label{eq:proof_thm_traj_eq_mult}
    \forall l_i \in \mathcal{L}^*_t, \forall T \in \mathcal{T}: Pr^+[l_i = l^*_t | l^*_t \in \mathcal{L}^*_t] \cdot Pr^-[l_i \leadsto T] \leq \\ \delta^{-1} \cdot e^\epsilon \cdot Pr^-[l_i = l^*_t] \cdot Pr^-[l_i \leadsto T]
\end{multline}
Since the output of $\mathcal{M}$ does not affect $Pr[l_i \leadsto T]$ and considering that $l_t = l^*_t$ and $l_i \leadsto T$ are independent events, then from Equation \ref{eq:proof_thm_traj_eq_mult}, we get:
\begin{multline}
    \forall l_i \in \mathcal{L}^*_t, \forall T \in \mathcal{T}: Pr^+[l_i = l^*_t \cap l_i \leadsto T | l^*_t \in \mathcal{L}^*_t] \leq \\ \delta^{-1} \cdot e^\epsilon \cdot Pr^-[l_i = l^*_t \cap l_i \leadsto T]
\end{multline}
We sum each side of the previous inequation for each $ l_i \in \mathcal{L}^*_t$ to get:
\begin{multline}
\label{eq:proof_thm_traj_eq_inter}
     \forall T \in \mathcal{T}: \sum\limits_{l_i \in \mathcal{L}^*_t} Pr^+[l_i = l^*_t \cap l_i \leadsto T | l^*_t \in \mathcal{L}^*_t] \leq \\ \delta^{-1} \cdot e^\epsilon \cdot \sum\limits_{l_i \in \mathcal{L}^*_t} Pr^-[l_i = l^*_t \cap l_i \leadsto T]
\end{multline}
Now, considering that $\forall \l_i \in \mathcal{L}_t\backslash \mathcal{L}^*_t: Pr^+[l_i=l^* \cap l_i \leadsto T | l^*_t \in \mathcal{L}^*_t] = 0$, then we can add $\sum_{\l_i \in \mathcal{L}_t\backslash \mathcal{L}^*_t} Pr^+[l_i=l^* \cap l_i \leadsto T | l^*_t \in \mathcal{L}^*_t]$ to the right side of Inequation \ref{eq:proof_thm_traj_eq_inter} to get:
\begin{multline}
     \forall T \in \mathcal{T}: \sum\limits_{l_i \in \mathcal{L}_t} Pr^+[l_i = l^*_t \cap l_i \leadsto T | l^*_t \in \mathcal{L}^*_t]  \leq \\ \delta^{-1} \cdot e^\epsilon \cdot \sum\limits_{l_i \in \mathcal{L}^*_t} Pr^-[l_i = l^*_t \cap l_i \leadsto T] 
\end{multline}
Then, we add $\sum_{l_i \in \mathcal{L}_t\backslash \mathcal{L}^*_t} Pr^-[l_i=l^* \cap l_i \leadsto T]$ to both sides of the previous inequation to get:
\begin{multline}
\label{eq:proof_thm_traj_eq_L_to_T}
     \forall T \in \mathcal{T}: \sum\limits_{l_i \in \mathcal{L}_t} Pr^+[l_i = l^*_t \cap l_i \leadsto T | l^*_t \in \mathcal{L}^*_t] + \\ \sum_{l_i \in \mathcal{L}_t\backslash \mathcal{L}^*_t} Pr^-[l_i=l^* \cap l_i \leadsto T] \leq \\ \delta^{-1} \cdot e^\epsilon \cdot \sum\limits_{l_i \in \mathcal{L}_t} Pr^-[l_i = l^*_t \cap l_i \leadsto T] 
\end{multline}
Since the probability distributions of $l_i = l*_t$ and $l_i \leadsto T$ are defined over the set $\mathcal{L}_t$, then we have:
\begin{equation}
    \sum_{l_i \in \mathcal{L}_t} Pr[l_i = l^* \cap l_i \leadsto T] = Pr[T = T^*]
\end{equation}
Now, we can use the previous equation to rewrite Inequation \ref{eq:proof_thm_traj_eq_L_to_T} as following:
\begin{multline}
\label{eq:proof_thm_traj_eq_T}
     \forall T \in \mathcal{T}: Pr^+[T = T^* | l^*_t \in \mathcal{L}^*_t] \leq \delta^{-1} \cdot e^\epsilon \cdot Pr^-[T = T^*] \\
     - \sum_{l_i \in \mathcal{L}_t\backslash \mathcal{L}^*_t} Pr^-[l_i=l^* \cap l_i \leadsto T]
\end{multline}
Considering the fact that $l_i=l^*$ and $l_i \leadsto T$ are independent events then we have:
\begin{multline}
\label{eq:proof_thm_traj_eq_last}
    -\sum_{l_i \in \mathcal{L}_t\backslash \mathcal{L}^*_t} Pr^-[l_i=l^* \cap l_i \leadsto T] \leq - \min_{l_i \in \mathcal{L}_t\backslash \mathcal{L}^*_t} Pr^-[l_i \leadsto T] ~ \cdot \\ \sum_{l_i \in \mathcal{L}_t\backslash \mathcal{L}^*_t} Pr^-[l_i=l^*] \\
    \leq \min_{l_i \in \mathcal{L}_t\backslash \mathcal{L}^*_t} Pr^-[l_i \leadsto T] ~ \cdot (\delta-1)
\end{multline}
Finally, based on Inequations \ref{eq:proof_thm_traj_eq_T} and \ref{eq:proof_thm_traj_eq_last}
we deduce:
\begin{equation}
    \forall T \in \mathcal{T}:  Pr^+[T = T^* | l^*_t \in \mathcal{L}^*_t]\leq \delta^{-1} \cdot e^\epsilon \cdot Pr^-[T = T^*] + \theta \nonumber
\end{equation}
with $\theta = \min_{l_i \in \mathcal{L}_t\backslash \mathcal{L}^*_t} Pr^-[l_i \leadsto T] ~ \cdot (\delta-1)$ which concludes the proof.

\subsection{Proof of Theorem \ref{thm:composition}}
\label{proof:composition}
    Proof is by induction. We first check that Theorem \ref{thm:composition} holds for $t=2$. In the rest of the proof, we use $Pr_{t}^-$ and $Pr_{t}^+$ to denote the prior and posterior (i.e., after the release of the timestep $t$ perturbed location)  probabilities. Since $\mathcal{M}_1$ satisfies ($\epsilon_1, \delta_1, \theta_1$)-trajectory privacy, then according to Definition \ref{def:trajecroty_privacy}, $\forall T \in \mathcal{T}$ we have:
    \begin{equation}
    \label{eq:thm_composition_t1}
        Pr^+_1[T = T^* | l^* \in \mathcal{L}^*_1] \leq \delta_1^{-1} \cdot e^{\epsilon_1} \cdot Pr^-_1[T = T^*] + \theta_1
    \end{equation}
    Now, considering that $\mathcal{M}_2$ satisfies ($\epsilon_2, \delta_2, \theta_2$)-trajectory privacy, then according to Definition \ref{def:trajecroty_privacy}, $\forall T \in \mathcal{T}$ we have:
    \begin{equation}
    \label{eq:thm_composition_t2}
        Pr^+_2[T = T^* | l^* \in \mathcal{L}^*_1] \leq \delta_2^{-1} \cdot e^{\epsilon_2} \cdot Pr^-_2[T = T^*] + \theta_2
    \end{equation}
    Since, no information are disclosed about the location of the user between the release of the perturbed locations $z_1$ and $z_2$, then $\forall T \in \mathcal{T}$ we have 
    \begin{equation}
    \label{eq:thm_composition_time_trans}
        Pr^+_1[T = T^* | l^* \in \mathcal{L}^*_1] = Pr^-_2[T = T^*]
    \end{equation}
    Hence, based on Equations \ref{eq:thm_composition_t1}, \ref{eq:thm_composition_t2}, and \ref{eq:thm_composition_time_trans}, we get:
    \begin{multline}
        Pr^+_2[T = T^* | l^* \in \mathcal{L}^*_2] \leq \delta_1^{-1} \cdot \delta_2^{-1} \cdot e^{\epsilon_1 + \epsilon_2} \cdot Pr^-_1[T = T^*] \\
        + \theta_2 + \theta_1 \cdot \delta_2^{-1} \cdot e^{\epsilon_2}
    \end{multline}
    which proves the theorem for $t=2$.
    
    Hence, to conclude the proof, we need to show that the theorem holds for $k+1$ while supposing that it holds for $k$. Considering the latter hypothesis, then $\forall T \in \mathcal{T}$ we have:
    \begin{multline}
        Pr^+_k[T = T^* | l^* \in \mathcal{L}^*_k] \leq \left(\prod_{i=1}^k \delta_i^{-1} \right) \cdot e^{\sum_{i=1}^k \epsilon_i} \cdot  Pr^-_1[T = T^*]\\
        + \theta_k + \sum_{i=1}^{k-1} \theta_i \prod_{j=i+1}^{k} \cdot \delta_{j}^{-1} \cdot e^{\epsilon_{j}}
    \end{multline}
    After releasing the the $k+1$ location using $\mathcal{M}_{k+1}$, $\forall T \in \mathcal{T}$ we have:
    \begin{multline}
        Pr^+_{k+1}[T = T^* | l^* \in \mathcal{L}^*_{k+1}] \leq \delta_{k+1} \cdot e^{\epsilon_{k+1}} \cdot \\ Pr^-_{k+1}[T = T^*]  + \theta_{k+1}
    \end{multline}
    since $Pr^+_k[T = T^* | l^* \in \mathcal{L}^*_k] = Pr^-_{k+1}[T = T^*]$, then we get:
    \begin{multline}
        Pr^+_{k+1}[T = T^* | l^* \in \mathcal{L}^*_{k+1}] \leq \left(\prod_{i=1}^{k+1} \delta_i^{-1}\right) \cdot e^{\sum_{i=1}^{k+1} \epsilon_i} \cdot Pr^-_1[T = T^*] \\
        + \theta_{k+1} + \sum_{i=1}^{k} \theta_i \prod_{j=i+1}^{k+1} \delta_{j}^{-1} \cdot e^{\epsilon_{j}}
    \end{multline}
    which concludes the proof.

\subsection{Proof of Theorem \ref{thm:link-LP-POIP}}
\label{proof:link-LP-POIP}
    The proof of the previous theorem follows the same reasoning as the proof of Theorem \ref{thm:link-LP-TP}.
    Since we suppose that the mechanism $\mathcal{M}$ satisfies ($\epsilon, \delta$)-location privacy, then according to Theorem \ref{thm:location-privacy}, we have:
\begin{equation}
    \forall l_i \in \mathcal{L}^*_t: Pr^+[l_i = l^*_t | l^*_t \in \mathcal{L}^*_t] \leq \delta^{-1} \cdot e^\epsilon \cdot Pr^-[l_i = l^*_t]
\end{equation}
We multiply both sides of the inequation with $Pr^-[l_i \leadsto P]$ to get
\begin{multline}
    \label{eq:proof_thm_traj_eq_mult_PIO}
    \forall l_i \in \mathcal{L}^*_t, \forall P \in \mathcal{P}: Pr^+[l_i = l^*_t | l^*_t \in \mathcal{L}^*_t] \cdot Pr^-[l_i \leadsto P] \leq \\ \delta^{-1} \cdot e^\epsilon \cdot Pr^-[l_i = l^*_t] \cdot Pr^-[l_i \leadsto P]
\end{multline}
Since the output of $\mathcal{M}$ does not affect $Pr[l_i \leadsto P]$, $l_t = l^*_t$ and $l_i \leadsto P$ are independant events. Then, based on Equation \ref{eq:proof_thm_traj_eq_mult_PIO}, we have:
\begin{multline}
    \forall l_i \in \mathcal{L}^*_t, \forall P \in \mathcal{P}: Pr^+[l_i = l^*_t \cap l_i \leadsto P | l^*_t \in \mathcal{L}^*_t] \leq \\ \delta^{-1} \cdot e^\epsilon \cdot Pr^-[l_i = l^*_t \cap l_i \leadsto P]
\end{multline}
We sum each side of the previous inequaion for each $ l_i \in \mathcal{L}^*_t$ to get:
\begin{multline}
\label{eq:proof_thm_traj_eq_inter_PIO}
     \forall P \in \mathcal{P}: \sum\limits_{l_i \in \mathcal{L}^*_t} Pr^+[l_i = l^*_t \cap l_i \leadsto P | l^*_t \in \mathcal{L}^*_t]  \leq \\ \delta^{-1} \cdot e^\epsilon \cdot \sum\limits_{l_i \in \mathcal{L}^*_t} Pr^-[l_i = l^*_t \cap l_i \leadsto P]
\end{multline}
Now, considering that $\forall \l_i \in \mathcal{L}_t\backslash \mathcal{L}^*_t: Pr^+[l_i=l^* \cap l_i \leadsto P | l^*_t \in \mathcal{L}^*_t] = 0$, then we can add $\sum_{\l_i \in \mathcal{L}_t\backslash \mathcal{L}^*_t} Pr^+[l_i=l^* \cap l_i \leadsto P | l^*_t \in \mathcal{L}^*_t]$ to the right side of Inequation \ref{eq:proof_thm_traj_eq_inter_PIO} to get:
\begin{multline}
     \forall P \in \mathcal{P}: \sum\limits_{l_i \in \mathcal{L}_t} Pr^+[l_i = l^*_t \cap l_i \leadsto P | l^*_t \in \mathcal{L}^*_t] \leq \\ \delta^{-1} \cdot e^\epsilon \cdot \sum\limits_{l_i \in \mathcal{L}^*_t} Pr^-[l_i = l^*_t \cap l_i \leadsto P]
\end{multline}
Then, we add $\sum_{l_i \in \mathcal{L}_t\backslash \mathcal{L}^*_t} Pr^-[l_i=l^* \cap l_i \leadsto P]$ to both sides of the previous inequation to get:
\begin{multline}
\label{eq:proof_thm_traj_eq_L_to_T_PIO}
     \forall P \in \mathcal{P}: 
     \sum\limits_{l_i \in \mathcal{L}_t} Pr^+[l_i = l^*_t \cap l_i \leadsto P | l^*_t \in \mathcal{L}^*_t] + \\ \sum_{l_i \in \mathcal{L}_t\backslash \mathcal{L}^*_t} Pr^-[l_i=l^* \cap l_i \leadsto P] \leq \\
     \delta^{-1} \cdot e^\epsilon \cdot \sum\limits_{l_i \in \mathcal{L}_t} Pr^-[l_i = l^*_t \cap l_i \leadsto P]
\end{multline}
Since the probability distibutions of $l_i = l*_t$ and $l_i \leadsto T$ are defined over the set $\mathcal{L}_t$, then we have:
\begin{equation}
    \sum_{l_i \in \mathcal{L}_t} Pr[l_i = l^* \cap l_i \leadsto P] = Pr[P = P^*]
\end{equation}
Now, we can use the previous equation to rewrite Inequation \ref{eq:proof_thm_traj_eq_L_to_T_PIO} as following:
\begin{multline}
\label{eq:proof_thm_traj_eq_T_PIO}
     \forall P \in \mathcal{P}: Pr^+[P = P^* | l^*_t \in \mathcal{L}^*_t] \leq \delta^{-1} \cdot e^\epsilon \cdot Pr^-[P = P^*] - \\ \sum_{l_i \in \mathcal{L}_t\backslash \mathcal{L}^*_t} Pr^-[l_i=l^* \cap l_i \leadsto P]
\end{multline}
Considering the fact that $l_i=l^*$ and $l_i \leadsto T$ are independent events, we get:
\begin{multline}
\label{eq:proof_thm_traj_eq_last_PIO}
    - \sum_{l_i \in \mathcal{L}_t\backslash \mathcal{L}^*_t} Pr^-[l_i=l^* \cap l_i \leadsto P] \leq - \min_{l_i \in \mathcal{L}_t\backslash \mathcal{L}^*_t} Pr^-[l_i \leadsto P] ~ \cdot \\ \sum_{l_i \in \mathcal{L}_t\backslash \mathcal{L}^*_t} Pr^-[l_i=l^*] \\
    \leq \min_{l_i \in \mathcal{L}_t\backslash \mathcal{L}^*_t} Pr^-[l_i \leadsto P] ~ \cdot (\delta-1)
\end{multline}
Finally, based on Inequations \ref{eq:proof_thm_traj_eq_T_PIO} and \ref{eq:proof_thm_traj_eq_last_PIO}
we deduce:
\begin{equation}
    \forall P \in \mathcal{P}: Pr^-[P = P^*] \leq \delta^{-1} \cdot e^\epsilon \cdot Pr^+[P = P^* | l^*_t \in \mathcal{L}^*_t] + \theta 
\end{equation}
which concludes the proof.
$\\~\\~\\~\\~\\~\\~\\~\\~\\~\\~\\~\\~\\~\\~\\~\\~\\~\\~\\~\\~$

\subsection{POI identification Algorithm}
\label{app:ident_poi}
\begin{algorithm}[H]
\caption{POI Extraction Algorithm}
\label{alg:poi_extraction}
\begin{algorithmic}[1]
\REQUIRE $\text{minTime} > 0$, $\text{maxDist} > 0$, $\text{minPts} \geq 1$, $|\text{data\_points}| > 0$
\STATE $\text{stay\_set} \leftarrow \emptyset$, $i \leftarrow 0$
\STATE $\text{candidate\_set} \leftarrow \emptyset$ \COMMENT{Ordered list of points}
\WHILE{$i < |\text{data\_points}|$}
    \STATE $\text{diameter} \leftarrow \max_{p \in \text{candidate\_set}} \text{distance}(\text{data\_points}[i], p)$
    \IF{$\text{diameter} \leq \text{maxDist}$}
        \STATE add $\text{data\_points}[i]$ to $\text{candidate\_set}$
        \STATE $i \leftarrow i + 1$
    \ELSE
        \IF{elapsed\_time in $\text{candidate\_set} \geq \text{minTime}$}
            \STATE add centroid of $\text{candidate\_set}$ to $\text{stay\_set}$
            \STATE $\text{candidate\_set} \leftarrow \emptyset$
        \ELSE
            \STATE remove first element of $\text{candidate\_set}$
        \ENDIF
    \ENDIF
\ENDWHILE
\IF{elapsed\_time in $\text{candidate\_set} \geq \text{minTime}$}
    \STATE add centroid of $\text{candidate\_set}$ to $\text{stay\_set}$
\ENDIF
\STATE $\text{clusters} \leftarrow \emptyset$
\FOR{$\text{stay}$ \textbf{in} $\text{stay\_set}$}
    \STATE $\text{neighborhood} \leftarrow \{s \in \text{stay\_set} \,|\, \text{dist}(s, \text{stay}) \leq \text{maxDist} \times 0.75\}$
    \IF{$|\text{neighborhood}| \geq \text{minPts}$}
        \FOR{$\text{cluster}$ \textbf{in} $\text{clusters}$}
            \IF{$\text{neighborhood} \cap \text{cluster} \neq \emptyset$}
                \STATE $\text{neighborhood} \leftarrow \text{neighborhood} \cup \text{cluster}$
                \STATE remove $\text{cluster}$ from $\text{clusters}$
            \ENDIF
        \ENDFOR
        \STATE add $\text{neighborhood}$ to $\text{clusters}$
    \ENDIF
\ENDFOR
\RETURN centroids of $\text{clusters}$
\end{algorithmic}
\end{algorithm}

\end{document}